\documentclass{article}
\usepackage[a4paper,margin=1in]{geometry}
\usepackage{booktabs} 
\usepackage[ruled]{algorithm2e} 

\SetAlFnt{\small}
\SetAlCapFnt{\small}
\SetAlCapNameFnt{\small}
\SetAlCapHSkip{0pt}
\IncMargin{-\parindent}

\usepackage{amsmath, amsthm, mathtools}
\usepackage{enumerate}
\usepackage{nicefrac}
\usepackage{bbm}
\usepackage{bm}
\usepackage{url}
\usepackage{thm-restate}
\usepackage{tikz, pgfplots}
\pgfplotsset{compat=1.18}

\usepackage{multirow}

\usepackage[suppress]{color-edits}
\definecolor{mygreen}{RGB}{20,180,20}
\addauthor{et}{blue}
\addauthor{er}{purple}
\addauthor{jk}{mygreen}

\DeclareMathOperator*{\argmax}{arg\,max} 
\usepackage{dsfont}
\newcommand{\1}[1]{\mathds{1}_{\left[#1\right]}}

\def\eg{{\em e.g.}}
\def\ie{{\em i.e.}}

\newtheorem{theorem}{Theorem}
\newtheorem{definition}[theorem]{Definition}

\usepackage[numbers]{natbib}
\setcitestyle{acmnumeric}

\title{Calibrated Recommendations for Users with Decaying Attention}

\author{
    Jon Kleinberg\thanks{Supported in part by a Vannevar Bush Faculty Fellowship, MURI grant W911NF-19-0217, AFOSR grant FA9550-19-1-0183, a Simons Collaboration grant, and a grant from the MacArthur Foundation.}\\
    Cornell University\\
    \texttt{kleinberg@cornell.edu}
    \and
    Emily Ryu\thanks{Supported in part by AFOSR grant FA9550-19-1-0183.}\\
    Cornell University\\
    \texttt{eryu@cs.cornell.edu}
    \and
    \'Eva Tardos\thanks{Supported in part by NSF grants CCF-1408673, CCF-1563714 and AFOSR grant FA9550-19-1-0183 and FA9550-23-1-0068}\\
    Cornell University\\
    \texttt{eva.tardos@cornell.edu}
}

\begin{document}

\maketitle

\begin{abstract}

There are many settings, including ranking and recommendation of content, where it is important to provide diverse sets of results,
with motivations ranging from fairness to novelty and other aspects of optimizing user experience. One form of diversity of recent interest is \emph{calibration}, the notion that personalized recommendations should reflect the full distribution of a user's interests, rather than a single predominant category --- for instance, a user who mainly reads entertainment news but also wants to keep up with news on the environment and the economy would prefer to see a mixture of these genres, not solely entertainment news. Existing work has formulated calibration as a subset selection problem; this line of work observes that the formulation requires the unrealistic assumption that all recommended items receive equal consideration from the user, but leaves as an open question the more realistic setting in which user attention decays as they move down the list of results. 

In this paper, we consider calibration with decaying user attention under two different models. In both models, there is a set of underlying genres that items can belong to. In the first setting, where items are coarsely binned into a single genre each, we surpass the $(1-1/e)$ barrier imposed by submodular maximization and provide a novel bin-packing analysis of a $2/3$-approximate greedy algorithm. In the second setting, where items are represented by fine-grained mixtures of genre percentages, we provide a $(1-1/e)$-approximation algorithm by extending techniques for constrained submodular optimization. Our work thus addresses the problem of capturing ordering effects due to decaying attention, allowing for the extension of near-optimal calibration from recommendation \emph{sets} to recommendation \emph{lists}.

\end{abstract}



\section{Introduction} \label{sec:intro}

Recommendation systems, now a ubiquitous feature of online platforms, have also been a long-standing source of fundamental theoretical problems in computing. Based on a model derived from a user's past behavior, 
such systems suggest relevant pieces of content that they predict the user 
is likely to be interested in. This is typically achieved by optimizing for an objective function based on a model of the user's interests (such as relevance or utility), and such questions lead to a number of interesting optimization questions.  Often, these basic formulations try to capture relevance in aggregate without considering the diversity of the results produced; they also generally treat lists of recommended results as an unordered sets, while in reality they are more accurately ordered sequences (reflecting the key influence of position and rank on the amount of attention a piece of content receives).
Considering these two directions in conjunction leads to new and interesting theoretical questions, which form the focus of this paper. 

In particular, a recurring concern with algorithmic recommendations is that
the process of optimizing for relevance risks producing results that
are too homogeneous; it can easily happen that all the most relevant
pieces of content are similar to one another, and that they collectively
correspond to only one facet of a user's interests at the expense of
other facets that go unrepresented~\cite{McNee2006}.
To address such concerns, a long-standing research paradigm seeks 
recommendation systems whose results 
are not only relevant but also {\em diverse}, reflecting the range
of a user's interests.
Explicitly pursuing diversity in recommendations has been seen 
as a way to help mitigate the homogenizing effects that might otherwise occur~\cite{agrawal2009,ashkan2015}.

\paragraph*{Calibrated recommendations}
Within this area, an active line of research has pursued
{\em calibration} as a means of optimizing for diversity \cite{Steck18}. 
In this formalism, we want to present a list of $k$ recommended items
to a single user (\eg, movies on an entertainment site,
or articles on a news site),
and there is a set of underlying {\em genres} that the items belong to.
The user has a {\em target distribution} over genres that
reflect the extent to which they want to consume each genre
in the long run.
A natural goal is that the average distribution induced by the
list of recommendations should be ``close,'' or \emph{calibrated}, to the user's target distribution.
(For example, a user who likes both documentaries and movies about
sports might well be dissatisfied with recommendations that were always
purely about sports and contained no documentaries;
this set of recommendations would be badly calibrated to the
user's target distribution of genres.)

In a user study that systematically varied the quality of results from a recommender system, the researchers reported significant differences in users' evaluations of the system based on the quality of results over a single session lasting only approximately 15 minutes \cite{userstudy}. Considering calibration an important aspect of quality (as asserted by multiple papers including \cite{cal2,kowald2023study}), it is thus critical to achieve calibration within each single session rather than to simply hope for recommendations to eventually ``average out'' in the long run, lest the user become so dissatisfied after a sufficiently miscalibrated session that they decide to abandon the system altogether.

Prior work by \cite{Steck18} showed that for 
natural measures of distributional similarity, the selection of a
set of $k$ items to match the user's target distribution can be
formulated as the maximization of a submodular set function.
Because of this, the natural greedy algorithm produces a set of $k$ items
whose distributional similarity to the user's target distribution
is within a $(1 - 1/e)$ factor of optimal.
In this way, the work provided an approximately optimal calibrated \emph{set} of recommendations.

\paragraph*{Decaying attention}
This same work observed a key limitation at the heart of
approximation algorithms for this and similar objectives:
it necessarily treats the $k$ recommendations as a set, for which
the order does not matter.
In contrast, one of the most well-studied empirical regularities
in the social sciences is decaying attention as a user reads through a list of results. Results at the top of a list get much more attention than
results further down --- this phenomenon has been documented not only through traditional content engagement metrics, but also directly through eye-tracking and other behavioral studies \cite{pan2007google,williamsZipf,NNGfold}.
Given this, the average genre distribution induced
by a list of $k$ recommendation results is really 
a {\em weighted average} over
the genres of these items, with the earlier items in the list
weighted more highly than the later ones.

Once we introduce the crucial property of
decaying user attention into the problem,
the formalism of set functions --- and hence of submodularity, which
applies to set functions --- is no longer available to us.
Moreover, it is no longer clear how to obtain algorithms for
provably near-optimal calibration.
There exist formalisms that extend the framework of submodular
functions, in restricted settings, to handle inputs that 
are ordered sequences \cite{Alaei19,Zhang16,Tschiatschek2017,Mitrovic18,Bernardini21,ec22seqsub}, but none of these formalisms can handle the setting of calibrated
recommendations with decaying attention that we have here.
It has thus remained an open question whether non-trivial approximation
guarantees can be obtained for this fundamental problem.

\paragraph*{The present work: Calibrated recommendations with decaying attention}
In this paper, we address this question by developing algorithms
that produce lists of recommendations with provably near-optimal
calibration for users with decaying attention.
We provide algorithms for two models of genres:
the \emph{discrete} model in which each item comes from a single genre, and
the \emph{distributional} model in which each item is described by a
distribution over genres. (For example, in this latter version,
a documentary about soccer in Italy is a multi-genre
mixture of a movie about sports, a movie about Italy, and a documentary.)

As noted above, a crucial ingredient in these models is to 
measure the similarity between two distributions: 
the user's desired target distribution over genres, and
the distribution of genres present in the results we show them.
\etreplace{In Section \ref{sec:calibration_overlap} w}{W}e make concrete what it means for these distributions to
be similar through the notion of an {\em overlap measure}, 
which we define in the paper to unify in a simple way
standard measures of distributional similarity.
Our results apply to a large collection of overlap measures 
including a large family of $f$-divergence measures with the 
property that similarities are always non-negative.
Overlap measures derived from the {\em Hellinger distance} are one
well-known measure in this family.
These were also at the heart of earlier
approaches that worked without decaying attention, where
these measures gave rise
to non-negative submodular set functions   
~\cite{calAMBM20,calNRAS22,calCD22}.\footnote{It is useful to note
that the KL-divergence --- arguably the other most widely-used
divergence along with the Hellinger distance --- is not naturally
suited to our problem, since it can take both positive and negative
values, and hence does not lead to well-posed questions about
multiplicative approximation guarantees.  This issue is not 
specific to models with decaying user attention;
the KL-divergence is similarly not well-suited to approximation
questions in the original unordered formalism, where the objective
function could be modeled as a set function.}

\paragraph*{Overview of results}
In our two \etreplace{models of genres}{genre models},\etcomment{I was afraid that two genre suggests that there are only 2 genres, maybe not a worry, but this is the phrase used earlier} we offer technical results of two distinct flavors \etedit{in Sections \ref{sec:discrete} and \ref{sec:distributional}}. First, the discrete genre model takes a completely new approach to the analysis of the greedy algorithm: to the best of our knowledge, our bin-packing argument is entirely novel; we also highlight that it allows us to surpass the barrier imposed by traditional submodularity arguments and achieve a stronger approximation guarantee.

For both versions of the problem, direct attempts at generalizing the methods of submodular maximization from unordered items to ordered items face a natural approximation barrier at $(1-1/e)$, simply because this is the strongest approximation guarantee we can obtain if we know only that the underlying function is submodular, and a special case of decaying attention is the case in which all weights are the same, which recovers the traditional submodular case. For the discrete version of the model, however, we are able to break through this $(1-1/e)$ barrier via a different technique based on a novel type of bin-packing analysis; through this approach, we are able to obtain a $2/3$-approximation to the optimal calibration for overlap measures based on the Hellinger distance. 
We find this intriguing, since the problem is NP-hard and
amenable to submodular maximization techniques; but unlike other
applications of submodular optimization (including hitting sets
and influence maximization) where $(1-1/e)$ represents the 
tight bound subject to hardness of approximation,
here it is possible to go further by using a greedy algorithm combined
with a careful analysis in place of submodular optimization.

To do this, we begin by observing that the objective function over 
the ordered sequence of items selected 
satisfies a natural inequality that can be viewed as
an analogue of submodularity, but for functions defined on 
sequences rather than on sets.
We refer to this inequality as defining a property that we call
{\em ordered submodularity}, and we show that ordered submodularity
by itself guarantees that the natural greedy algorithm for
sequence selection (with repeated elements allowed, as in our
discrete problem) provides a $1/2$-approximation to the optimal sequence.

This bound of $1/2$ is not as strong as $1-1/e$; but unlike the
techniques leading to the $1-1/e$ bound, the bound coming
from ordered submodularity provides a direction along which we
are able to obtain an improvement.
In particular, for the discrete problem we can think of 
each genre as a kind of ``bin'' that contains items belonging
to this genre, and the problem of approximating a desired target
distribution with respect to the Hellinger measure then becomes
a novel kind of load-balancing problem across these bins.
Using a delicate local-search analysis, we are able to maintain
a set of inductive invariants over the execution of a greedy
bin-packing algorithm for this problem and show that it satisfies
a strict strengthening of the general ordered submodular inequality;
and from this, we are able to show that it maintains 
a $2/3$-approximation bound.

Subsequently, \etedit{in Section \ref{sec:distributional} for} the distributional genre model \etreplace{we build}{builds} on an existing line of work on constrained submodular maximization by introducing a new transformation technique to allow for position-based weights, which were not previously handled. A separate line of work has posed, but left open, the question of the effect of such position-based weights on achieving near-optimal diversity in recommender systems. Our work unites these two bodies of research by developing new methods from the former line of work to answer questions from the latter, and thereby provide a deeper fundamental understanding of the effects of weights and ordering on approximate submodular maximization. 

For the case of distributional genres, 
we begin by noting that if we were to make the unrealistic assumption
of repeated items (i.e. availability of many items with the exact same genre
distribution $q$), then we could apply a form of submodular optimization
with matroid constraints of \cite{calinescu2011} to obtain a $(1-1/e)$-approximation
to the optimal calibration with decaying attention.
This approach is not available to us, however, when we make the
more reasonable assumption that items each have their own specific genre distribution.
Instead, we construct a more complex laminar matroid structure, 
and we are able to show that with these more complex constraints,
a continuous greedy algorithm and pipage rounding produces a sequence of items within $(1 - 1/e)$ of optimal.




\section{Related Work} \label{sec:relatedwork}
The problem of calibrated recommendations was defined by \cite{Steck18}, in which \emph{calibration} is proposed as a new form of diversity with the goal of creating recommendations that represent a user's interests. In this model, items represent distributions over genres, and weighting each item's distribution according to its rank induces a genre distribution for the entire recommendation list. Calibration is then measured using a \emph{maximum marginal relevance} objective function,  a modification of the KL divergence
from this induced distribution to the user's desired distribution of interests. In the case where all items are weighted equally, the maximum marginal relevance function is shown to be monotone and submodular, and thus $(1-1/e)$-approximable by the standard greedy algorithm. However, when items have unequal weights (such as with decaying user attention), the function becomes a sequence function rather than a set function, and the tools of submodular optimization can no longer be applied. Further, the use of
KL divergence with varying weights results in a mixed-sign objective function (refer to Appendix~\ref{sec:bad_netflix} for an example), meaning that formal approximation guarantees are not even technically well-defined in this setting. Hence, \cite{Steck18}'s approximation results are limited to only the equally-weighted (and therefore essentially unordered) case. 

Since then, there has been much recent interest in improving calibration in recommendation systems, via methods such as greedy selection using statistical divergences directly or other proposed calibration metrics~\cite{calNRAS22,calCD22} and LP-based heuristics~\cite{calSAM21}. However, this line of work largely focuses on empirical evaluation of calibration heuristics rather than approximation algorithms for provably well-calibrated lists. To the best of our knowledge, our work provides the first nontrivial approximation guarantees for calibration with unequal weights due to decaying attention. 

Within the recommendation system literature, there is a long history of modeling calibration and other diversity metrics as submodular set functions, and leaving open the versions where ordering matters because user engagement decays over the course of a list (\eg, \cite{agrawal2009,ashkan2015,Steck18}). Although numerous approaches to extending the notion of submodularity to have sequences have been proposed (\eg, \cite{Alaei19,Zhang16,Tschiatschek2017,Mitrovic18,Bernardini21,ec22seqsub,zhang2022ranking,udwani21order}), none is designed to handle these types of ordering effects. For a detailed survey of general theories of submodularity in sequences and a discussion of how they do not model our problem of calibration with decaying user attention, we refer the reader to Appendix~\ref{sec:extended_survey}.



\section{Problem Statement and Overlap Measures} \label{sec:calibration_overlap}
\cite{Steck18} considers the problem of creating calibrated recommendations using the language of \emph{movies} as the items with which users interact, and \emph{genres} as the classes of items. Each user has a preference distribution over genres that can be inferred from their previous activity, and the goal is to recommend a list of movies whose genres reflect these preferences (possibly also incorporating a ``quality'' score for each movie, representing its general utility or relevance). In our work, we adopt \cite{Steck18}'s formulation of distributions over genres and refer to items as movies (although the problem of calibrated recommendations is indeed more general, including also news articles and other items, as discussed in the introduction). We describe the formal definition of our problem next.

\subsection{Item Genres and Genre of Recommendation Lists} \label{sec:genre_notation}
Consider a list of recommendations $\pi$ for a user $u$.
Let $p(g)$ be the distribution over genres $g$ preferred by the user (possibly inferred from previous history). Given our focus on a single user $u$, we keep the identity of the user implicit in the notation. For simplicity of notation, we will label the items as the elements of $[K]$, and say that item $i$ has genre distribution $q_i$. 
    
Following the formulation of \cite{Steck18}, we define the distribution over genres $q(\pi)$ of a recommendation list $\pi = \pi_1 \pi_2 \dots \pi_k$ as
    $\displaystyle{
    q(\pi)(g) \coloneqq 
    \sum_{j=1}^k w_j  \cdot q_{\pi_j}(g),
    }$
    where $w_j$ is the weight of the movie in position $j$, and we assume that the weights sum to 1: $\sum_{j=1}^k w_j = 1$.\footnote{Various weighting schemes are possible; \cite{Steck18} suggests that “Possible choices include the weighting schemes used in ranking metrics, like in Mean Reciprocal Rank (MRR) or normalized Discounted Cumulative Gain (nDCG).” Alternatively, given empirical measurements of attention decay such as in~\cite{pan2007google}, one might use numerically estimated weights.}
    Note that the \emph{position-based weights} make the position of each recommendation important, so this is no longer a subset selection problem.

To model attention decay, we will assume that the weights are weakly decreasing in rank (i.e., $w_a \ge w_b$ if $a < b$). We also assume that the desired length of the recommendation list is a fixed constant $k$. This assumption is without loss of generality, even with the more typical cardinality constraint that the list may have length \emph{at most} $k$ --- we simply consider each possible length $\ell \in [1,k]$, renormalize so that the first $\ell$ weights sum to $1$, and perform the optimization. We then take the maximally calibrated list over all $k$ length-optimal lists. 

The goal of the \emph{calibrated recommendations} problem is to choose $\pi$ such that $q(\pi)$ is ``close'' to $p$. To quantify closeness between distributions, we introduce the formalism of \emph{overlap measures}.

\subsection{Overlap Measures}
For the discussion that follows, we restrict to finite discrete probability spaces $\Omega$ for simplicity, although the concepts can be generalized to continuous probability measures.

A common tool for quantitatively comparing distributions is statistical divergences, which measure the ``distance'' from one distribution to another. A divergence $D$ has the property that $D(p,q) \ge 0$ for any two distributions $p,q$, with equality attained if and only if $p=q$. This means that divergences cannot directly be used to measure calibration, which we think of as a non-negative metric that is uniquely \emph{maximized} when $p=q$. Instead, we define a new but closely related tool that we call \emph{overlap}, which exactly satisfies the desired properties. 

Our definition is also more general in two important ways. First, we do not limit ourselves to the KL divergence, so that other divergences and distances with useful properties may be used (such as the Hellinger distance, $H(p,q) = \frac{1}{\sqrt{2}} ||\sqrt{p} - \sqrt{q} ||_2$, which forms a bounded metric and has a convenient geometric interpretation using Euclidean distance). Second, in our definition $q$ may be any \emph{subdistribution}, a vector of probabilities summing to \emph{at most} $1$. This is crucial because it enables the use of algorithmic tools such as the greedy algorithm -- which incrementally constructs $q$ from the $0$ vector by adding a new movie (weighted by its rank), and thus in each iteration must compute the overlap between the true distribution $p$ and the partially constructed subdistribution $q$.

\begin{definition}[Overlap measure]\label{def:overlap}
An \textbf{overlap measure} $G$ is a function on pairs of distributions and subdistributions $(p,q)$ with the properties that 
\begin{enumerate}[(i)]
    \item $G(p,q) \ge 0$ for all distributions $p$ and subdistributions $q$,
    \item for any fixed $p$, $G(p,q)$ is uniquely maximized at $q=p$.
\end{enumerate}
\end{definition}

Further, we observe that overlap measures can be constructed based on distance functions.

\begin{definition}[Distance-based overlap measure]
Let $d(p,q)$ be a bounded distance function on the space of distributions $p$ and subdistributions $q$ with the property that $d(p,q) \ge 0$, with $d(p,q) = 0$ if and only if $p=q$. Denote by $d^*$ the maximum value attained by $d$ over all pairs $(p,q)$. Then, the \textbf{$d$-overlap measure} $G_d$ is defined as $G_d(p,q) \coloneqq d^* - d(p,q).$
\end{definition}
Now, it is clear that $G_d$ indeed satisfies both properties of an overlap measure (Definition \ref{def:overlap}): property (i) follows from the definition of $d^*$, and property (ii) follows from the unique minimization of $d$ at $q=p$. 

For an overlap measure $G$ and a recommendation list $\pi = \pi_1 \pi_2 \dots \pi_k$, we define $G(\pi) \coloneqq G(p,q(\pi))=G(p, \sum_{i=1}^k w_i q_{\pi_i})$.

\subsection{Constructing Families of Overlap Measures} \label{sec:overlap_families}

An important class of distances between distributions are $f$-divergences. Given a convex function $f$ with $f(1) = 0$, the $f$-divergence from distribution $q$ to distribution $p$ is 
$
D_f(p,q) \coloneqq \sum_{x\in \Omega} f\left( \frac{p(x)}{q(x)}\right) q(x).
$
One such $f$-divergence is the KL divergence, which \cite{Steck18} uses to define a \emph{maximum marginal relevance} objective function similar to an overlap measure. However, this proposed function has issues with mixed sign (see Appendix~\ref{sec:bad_netflix} for an example),
so it does not admit well-specified formal approximation guarantees. Instead, we consider a broad class of overlap measures based on $f$-divergences for all convex functions $f$. As a concrete example, consider the squared Hellinger distance (obtained by choosing $f(t) = (\sqrt{t}-1)^2$ or $f(t) = 2(1-\sqrt{t})$), which is of the form $$H^2(p,q) = \frac{1}{2} \sum_{x\in \Omega} (\sqrt{p(x)} - \sqrt{q(x)})^2 = 1 - \sum_{x\in \Omega} \sqrt{p(x) \cdot q(x)}.$$ This divergence is bounded above by $d^* = 1$; the resulting $H^2$-overlap measure is $$G_{H^2} (p,q) = \sum_{x\in \Omega} \sqrt{p(x) \cdot q(x)}.$$ 

Inspired by the squared Hellinger-based overlap measure, we also construct another general family of overlap measures based on non-decreasing concave functions. Given any nonnegative non-decreasing concave function $h$, we define the overlap measure $G^h(p,q) = \sum_{x\in \Omega} \frac{h(q(x))}{h'(p(x))}$. For instance, taking $h(x) = x^\beta$ for $\beta \in (0,1)$ gives $\frac{1}{h'(x)} = \frac{1}{\beta} x^{1-\beta}$, which produces the (scaled) overlap measure $G^{x^\beta}(p,q) = \sum_{x\in \Omega} p(x)^{1-\beta} q(x)^\beta.$ 
Observe that the natural special case of $\beta = \frac{1}{2}$ gives $h(x) = \frac{1}{h'(x)} = \sqrt{x}$, providing an alternate construction that recovers the squared Hellinger-based overlap measure.  

\subsection{Monotone Diminishing Return (MDR) Overlap Measures} \label{sec:mdr_overlap}

Many classical distances, including those discussed above, are originally defined on pairs of distributions $(p,q)$ but admit explicit functional forms that can be evaluated using the values of $p(x)$ and $q(x)$ for all $x \in \Omega$. This allows us to compute $d(p,q)$, and consequently $G_d(p,q)$, when $q$ is not a distribution (\ie, the values do not sum to $1$), which will be useful in defining algorithms for finding well-calibrated lists. Using this extension, we can take advantage of powerful techniques from the classical submodular optimization literature when a certain extension of the overlap measure $G_d$ is monotone and submodular, properties satisfied by most distance-based overlap measures. 

Consider an extension of an overlap measure $G(p,q)$ to a function on the ground set $V=\{(i,j)\}$, where $i\in[K]$ is an item and $j\in [k]$ is a position. For a set $R \subseteq V$, define $R^{\le j}$ be the set of items assigned to position $j$ or earlier; that is, 
$$
R^{\le j} \coloneqq \left\{ i \in [K] \mid  \exists \ell\le j \text{ s.t. } (i,\ell) \in R \right\}.
$$
Assuming the overlap measure $G$ is well-defined as long as the input $q$ is non-negative (but not necessarily a probability distribution), we define the set function 
 $$F_G(R) \coloneqq G\left(p, \sum_{j=1}^{k} w_j \left(\sum_{i\in R^{\le j}\setminus R^{\le j-1} } q_i\right)\right).$$

With this definition, we can define monotone diminishing return (MDR) and strongly monotone diminishing return (SMDR) overlap measures:
\begin{definition}[(S)MDR overlap measure]
An overlap measure $G$ is \textbf{monotone diminishing return (MDR)} if its corresponding set function $F_G$ is monotone and submodular. If, in addition, $G$ is non-decreasing with respect to all $q(x)$, we say $G$ is \textbf{strongly monotone diminishing return (SMDR)}.
\end{definition}

For any bounded monotone $f$-divergence, the corresponding overlap measure satisfies the MDR property, where by monotone we mean that if subdistribution $q_2$ coordinate-wise dominates subdistribution $q_1$, then $D_f(p,q_1) \ge D_f(p,q_2)$ for all $p$. Further, all overlap measures $G^h$ defined above by concave functions $h$ satisfy the SMDR property. A detailed technical discussion is deferred to Appendix~\ref{sec:proofsMDR}, but at a high level, since $D_f$ is negated in the construction of $G_{D_f}$, the convexity of $f$ (since $D_f$ is negated) and the concavity of $h$ result in concave overlap measures (corresponding to diminishing returns).

\begin{restatable}{theorem}{fdivergence}
\label{thm:overlap_fdiv}
Given any bounded monotone $f$-divergence $D_f$ with maximum value $d^* = \max_{(p,q)} D_f(p,q)$, the corresponding $D_f$-overlap measure $G_{D_f}(p,q) = d^* - D_f(p,q)$ is MDR.
\end{restatable}

\begin{restatable}{theorem}{hoverlap}
\label{thm:overlap_h}
Given any nonnegative non-decreasing concave function $h$, the overlap measure $G^h(p,q) = \sum_{x\in \Omega} \frac{h(q(x))}{h'(p(x))}$ is SMDR.
\end{restatable}

Observe that $D_f$-overlap measures are not necessarily \emph{SMDR}, but many $D_f$-overlap measures based on commonly used $f$-divergences are not only bounded and monotone, but also increasing in $q(x)$ -- this includes the squared Hellinger distance defined above, so the resulting overlap measure $G_{H^2}(p,q)$ is indeed both MDR and SMDR.



\section{Calibration in the Discrete Genre Model} \label{sec:discrete}

In this section we consider the version of the calibration
model with \emph{discrete genres}, 
in which each item is classified into a single genre. 
In this model, we allow the list of items to contain repeated genres, since it is natural to assume that the universe contains many items of each genre, and that a recommendation list may display multiple items of the same genre.

We start by thinking about a solution to the problem in this model as a sequence of choices of genres, and we study how the value of
the objective function changes as we append items to the
end of the sequence being constructed.
In particular, we show that as we append items,
the value of the objective function
changes in a way that is
governed by a basic inequality, that intuitively can be viewed as an analogue of monotonicity and submodularity but for sequences rather than sets.
We pursue this idea by defining any function on sequences
to be {\em ordered-submodular} if it satisfies this basic
inequality; in particular, the Hellinger measures of calibration
for our problem (as well as more general families based on the
overlap measures defined earlier) are ordered-submodular in this sense.

As a warm-up to the main result of this section, we start 
by showing that for {\em any} ordered-submodular function, the natural greedy algorithm that iteratively adds items to maximally increase the objective function achieves 
a factor $1/2$-approximation to the optimal sequence.
Note that this approximation guarantee is weaker than the $(1-1/e)$ guarantee obtained by classical (constrained) submodular optimization, 
but we present it because it creates a foundation for
analyzing the greedy algorithm which we can then strengthen
to break through the $(1 - 1/e)$ barrier and achieve a
$2/3$-approximation for the problem of calibration with discrete genres.
(In contrast, the techniques achieving $1 - 1/e$ appear to be
harder to use as a starting point for improvements, since they
run up against tight hardness bounds for submodular maximization.) 

To start, we make precise exactly how the greedy algorithm works
for approximate maximization of a function $f$ over sequences. 
The greedy algorithm
initializes $A_0 = \emptyset$ (the empty sequence), and for $\ell = 1, 2, \dots, k$, it selects $A_\ell$ to be the sequence that maximizes our function $f(A)$ over all sequences obtained by appending an element to the end of $A_{\ell-1}$. In other words, it iteratively appends elements to the sequence $A$ one by one, each time choosing the element that leads to the greatest marginal increase in the value of $f$.

To simplify notation, for two sequences $A$ and $B$ we use $A||B$ to denote their concatenation. For a single element $s$, we use $A||s$ to denote $s$ added at the end of the list $A$.

\subsection{Ordered-submodular Functions and the Greedy Algorithm} \label{sec:ordsub}

Let $f$ be a function defined on a sequences of elements from some ground set; we say that $f$ is \emph{ordered-submodular} if for all sequences of elements $s_1 s_2 \dots s_k$, the following property holds for all $i \in [k]$ and all other elements $\bar{s}_i$: 
\begin{equation}\label{eq:ordsub}
    f(s_1 \dots s_i) - f(s_1 \dots s_{i-1}) \ge f(s_1 \dots s_i \dots s_k) - f(s_1 \dots s_{i-1} \bar{s}_i s_{i+1} \dots s_k).
\end{equation}

Notice that if $f$ is an ordered-submodular function that takes
sequences as input but does not depend on their order
(that is, it produces the same value for all permutations of
a given sequence), then it follows immediately from the definition
that $f$ is a monotone submodular set function.
In this way, monotone submodular set functions are a special
case of our class of functions.

A standard algorithmic inductive argument shows that the greedy algorithm described earlier attains a $1/2$-approximation to the optimal sequence. Next, observe that the MDR property defined in Section \ref{sec:overlap_families} directly implies ordered submodularity (via submodularity and monotonicity of $\hat{F}_G$), and hence the greedy algorithm is a $1/2$-approximation algorithm for these calibration problems. Full proofs of both claims above as well as the theorem are given in Appendix~\ref{sec:proofsordsub}.

\begin{restatable}{theorem}{proptwoapprox}
\label{prop:2approx} 
The greedy algorithm for nonnegative ordered-submodular function maximization over sets of cardinality $k$ outputs a solution whose value is at least $\frac{1}{2}$ times that of the optimum solution.
\end{restatable}

\begin{restatable}{theorem}{ordsuboverlap}
\label{thm:ordsub_overlap}
Any MDR overlap measure $G$ is ordered-submodular. Thus, the greedy algorithm provides a $1/2$-approximation for calibration heuristics using MDR overlap measures.
\end{restatable}

\subsection{Improved Approximation for Calibration with Discrete Genres} \label{sec:cal2/3}

Next, we focus on calibration using the squared Hellinger-based overlap measure, which has several useful properties: (1) it is SMDR, and thus the approximation guarantee is directly comparable to the $(1-1/e)$ guarantee in the distributional model that we discuss next in Section~\ref{sec:distributional}; (2) its mathematical formula is amenable to genre-specific manipulations; (3) perhaps most importantly, it is well-motivated by frequent use in the calibration literature (\eg, \cite{calAMBM20,calNRAS22,calCD22}). (We note that our techniques apply generally to many overlap measures, such as the second family based on concave functions described in Section~\ref{sec:overlap_families}, but the quantitative $2/3$ bound is specific to the squared Hellinger-based overlap measure.\footnote{In particular, our bin-packing analysis of the greedy algorithm relies on concavity along the direction of improvement, so it applies to other overlap measures such as those in the $G^{x^\beta}$ family, but the numerical constant of $\frac12$ in Lemma~\ref{lem:ineq1/2} (and thus the final approximation guarantee of $\frac23$ in Theorem~\ref{thm:2/3approx}) would change. Here, we focus on the particular case of $\beta = \frac12$, as the induced Hellinger-based overlap measure is one that is commonplace in practice.}) We prove this improved approximation result using the concrete form of the Hellinger distance to establish a stronger version of the ordered submodularity property.

Given that each item belongs to a single genre, and that we have many copies of items for each genre, the question we ask in each step of the greedy algorithm is now: at step $i$, which genre should we choose to assign weight $w_i$ to? We can think of this problem as a form of ``bin-packing'' problem, packing the weight $w_i$ into a bin corresponding to a genre $g$. 

Since every item represents a single discrete genre, we can interpret a recommendation list as an assignment of \emph{slots} to \emph{genres}.  Then, using $s_i = g$ to denote that a sequence $S$ assigns slot $i$ to genre $g$, we can write the squared Hellinger-based overlap measure as
\begin{equation}
\label{eq:discrete-Hellinger}
f(S) = \sum_{\text{genres } g} \sqrt{p(g)} \sqrt{\sum_{i\in [k] : s_i = g} w_i}.  
\end{equation}

The main technical way we rely on the Hellinger distance is the following Lemma, which strengthens inequality (\ref{eq:approx-ineq}) but does not assume that the sequence $T^i$ or $\bar T^{i+1}$ is coming from the optimal sequence, or that they are identical except for their first element.

\begin{restatable}{lemma}{strongerineq}
\label{lem:ineq1/2} With calibration defined via the Hellinger distance, for all sequences $A_{i-1}$ and $T^i$, and the greedy choice of extending $A_{i-1}$ with the next element $a_{i}$, there exists a sequence $\bar T^{i+1}$ such that 
$$f(A_{i} || \bar T^{i+1}) \ge f(A_{i-1} || T^i) - \frac{1}{2} \left(f(A_i) - f(A_{i-1}) \right).$$
\end{restatable}

Before we prove the lemma, we show that it inductively yields a $2/3$-approximation guarantee: 

\begin{restatable}{theorem}{twothirds}\label{thm:2/3approx}
For the calibration problem with discrete genres, the greedy algorithm provides a $2/3$-approximation for the squared Hellinger-based overlap measure.
\end{restatable}
\begin{proof}
We define $S^{(1)}$ to be the optimal sequence $S$, and using Lemma~\ref{lem:ineq1/2} with $T^i = S^{(i)}$, we define inductively $S^{(i+1)} = \bar T^{i+1}$ (the existence of which is stated by the lemma).

We show via induction that for all $i$, $f(A_i || S^{(i+1)}) \ge OPT(k) - \frac{1}{2}f(A_i).$

    For the base case of $i=0$, we have $f(A_0 || S^{(1)}) = f(S) = OPT(k) \ge OPT(k) - \frac{1}{2}f(A_0)$. So suppose the claim holds for some $i$, and observe that by Lemma~\ref{lem:ineq1/2} and the fact that $f(A_{i+1}) \ge f(A_i || s_{i+1})$ by definition of the greedy algorithm, we have
    \begin{align*}
        f(A_{i+1} || S^{(i+2)}) &\ge f(A_i || S^{(i+1)}) - \frac{1}{2} (f(A_{i+1}) - f(A_{i})) \\
        &\ge OPT(k) - \frac{1}{2} f(A_i) - \frac{1}{2} (f(A_{i+1}) - f(A_{i})) \\
        &= OPT(k) - \frac{1}{2} f(A_{i+1}),
    \end{align*}
    completing the induction. Finally, setting $i=k$ establishes $ALG(k) \ge \frac{2}{3} OPT(k)$.
\end{proof}

\noindent \textbf{Remark.} We note that this approximation guarantee is fairly robust to settings in which we do not have perfectly accurately information about the preferences and genres, but only with a small degree of error or noise to within a multiplicative factor of $(1+\varepsilon)$. In this case, we still maintain a $\frac{2/3}{(1+\varepsilon)^2}$-approximation; for more details, see Appendix~\ref{sec:noisy_approx}.

Next, we outline the proof of Lemma~\ref{lem:ineq1/2}; the full analysis is in Appendix~\ref{sec:proofscal2/3}.

\begin{proof}[Proof outline of Lemma \ref{lem:ineq1/2}]
Consider the sequence $A_{i-1} || T^i$ and the greedy choice $a_i$, and let $t_i$ be the first element of $T^i$. Recall that each of these items is a genre, and the term multiplying $\sqrt{p(g)}$ in the Hellinger distance (\ref{eq:discrete-Hellinger}) is the sum of the weights of all positions where a given genre $g$ is used. To show the improved bound, it will help to define notation for the total weight of positions that have a genre $g$ in $A_{i-1}$ and in $T^{i+1}$ respectively, skipping the genre in the $i^\text{th}$ position. Since this lemma focuses on a single position $i$, we will keep $i$ implicit in some of the notation.

Let $\alpha(g) \coloneqq \sum_{\substack{\{j\in [i-1], a_j = g\}}} w_j$ denote the total weight of the slots assigned to genre $g$ by $A_{i-1}$. Let $\tau(g) \coloneqq \sum_{\substack{\{j\in [i+1, k], t_j = g\}}} w_j$ denote the total weight assigned to genre $g$ by $T^{i+1}$. Say that the greedy algorithm assigns slot $i$ to genre $a_i=g'$, but in $T^i$ the first genre (corresponding to slot $i$) is $t_i=g^*$. 

Next, notice that for the squared Hellinger-based overlap measure (\ref{eq:discrete-Hellinger}), there are only two genres in which $f(A_{i} || T^{i+1})$  and $ f(A_{i-1} || T^i)$ differ: the genre $a_i=g'$ chosen by the greedy algorithm, and the genre $t_i = g^*$ of the first item of the sequence $T^i$. For all other genres, the sum of assigned weights in the definition of the Hellinger distance is unchanged.

First, writing $T^{i+1}$ to denote simply dropping the first item assignment from $T^i$, and denoting the blank in position $i$ by $\_$, we get

\begin{align*}
    f(A_{i-1} || a_i || T^{i+1}) - f(A_{i-1} || \_ || T^{i+1}) & = \sqrt{p(g')} \left(\sqrt{\alpha(g') + w_i + \tau(g')} - \sqrt{\alpha(g') + \tau(g')} \right), \\
    f(A_{i-1} || t_i) - f(A_{i-1}) &= \sqrt{p(g^*)} \left(\sqrt{\alpha(g^*) + w_i} - \sqrt{\alpha(g^*)} \right), \\
    f(A_{i-1} || t_i || T^{i+1}) - f(A_{i-1} || \_ || T^{i+1}) &= \sqrt{p(g^*)} \left(\sqrt{\alpha(g^*) + w_i + \tau(g^*)} - \sqrt{\alpha(g^*) + \tau(g^*)} \right). 
\end{align*}

Using these expressions and the monotonicity and convexity of the square root function, we get 
\begin{align*}
    f(A_{i-1} || T^i) - f(A_i || T^{i+1}) &= \sqrt{p(g^*)} \left(\sqrt{\alpha(g^*) + w_i + \tau(g^*)} - \sqrt{\alpha(g^*) + \tau(g^*)} \right) \\
    &\qquad - \sqrt{p(g')} \left(\sqrt{\alpha(g') + w_i + \tau(g')} - \sqrt{\alpha(g') + \tau(g')} \right) \\
    &\le \sqrt{p(g^*)} \left(\sqrt{\alpha(g^*) + w_i + \tau(g^*)} - \sqrt{\alpha(g^*) + \tau(g^*)} \right) \\
    &\le \sqrt{p(g^*)} \left(\sqrt{\alpha(g^*) + w_i} - \sqrt{\alpha(g^*)} \right) \\
    &= f(A_{i-1} || t_i) - f(A_{i-1}) \\
    &\le f(A_i) - f(A_{i-1})
\end{align*}

To obtain the improved bound, we need to distinguish a few cases.
If $f(A_i || T^{i+1}) \ge f(A_{i-1} || T^i)$ (\eg, if $g'=g^*$), then it suffices to take $\bar T^{i+1} = T^{i+1}$ and the inequality holds trivially. Hence, we assume that $g' \ne g^*$ and $f(A_i || T^{i+1}) \le f(A_{i-1} || T^i)$.

Now, we may need to modify $T^i$ to get $\bar T^{i+1}$, depending on the size of $\tau(g')$ relative to $w_i$.

\textbf{Case 1: $\tau(g') \ge \frac{1}{2} w_i$.} 
Intuitively, since the greedy algorithm added $w_i$ to $g'$ instead of $g^*$, we should not assign so much additional weight to $g'$. To create $\bar T^{i+1}$, we start from $T^{i+1}$ (the part of $T^i$ without the first item), but make an improvement by reassigning some subsequent items from $g'$ to $g^*$.

Because $\tau(g')$ is the sum of weights each of which is at most $w_i$ (as the weights of positions are in decreasing order), we can move some weight $z$ satisfying $\frac{1}{2}w_i \le z \le w_i$ from $g'$ to $g^*$.  Now, consider the function 
$$c(x) = \sqrt{p(g')} \sqrt{\alpha(g') + \tau(g') + w_i - x} + \sqrt{p(g^*)} \sqrt{\alpha(g^*) + \tau(g^*) + x},$$
representing the contribution from genres $g'$ and $g^*$ towards $f$, after a contribution that moves an amount $x$ from $g'$ to $g^*$ (the change in $f$ will only be due to these two genres, since all others are unchanged). Observe that $x=0$ corresponds to $f(A_i || T^{i+1})$ and $x=w_i$ corresponds to $f(A_{i-1} || T^i)$, so $c(0) \le c(w_i)$. Further, $c$ is concave in $x$. As depicted in the figure below, a correction that is at least $\frac{1}{2} w_i$ increases $f$ by at least half the amount that a full correction of $w_i$ would have achieved. 

\begin{figure}[ht] \label{fig:concave}
    \begin{tikzpicture}[scale=0.8]
    \begin{axis}[
            ticks=none,
            axis lines=middle,
            xlabel={$x$},
            xmin=0, xmax=9,
            ymin=0, ymax=4.2,
            clip=false,
            xticklabels=\empty,
            yticklabels=\empty
        ]
        \addplot+[mark=none,samples=200,unbounded coords=jump, domain=0.5:8.5] {ln(x)+1.5};

        \draw[fill] (axis cs:1,1.5) circle [radius=1.5pt] node[below right] {};
        \draw[fill] (axis cs:4,2.8863) circle [radius=1.5pt] node[below right] {};
        \draw[fill] (axis cs:7,3.4459) circle [radius=1.5pt] node[below right] {};

        \draw (axis cs:1,0) node[below] {$0$};
        \draw (axis cs:4,0) node[below] {$z \ge \nicefrac{w_i}{2}$};
        \draw (axis cs:7,0) node[below] {$w_i$};
        \draw (axis cs:0,1.5) node[left] {$f(A_{i} || T^{i+1})$};
        \draw (axis cs:0,2.8863) node[left] {$f(A_{i} || \bar T^{i+1})$};
        \draw (axis cs:0,3.4459) node[left] {$f(A_{i-1} || T^i)$};
        
        \draw[dashed] (axis cs:1,0) -- (axis cs:1,1.5) -- (axis cs:0,1.5);
        \draw[dashed] (axis cs:4,0) -- (axis cs:4,2.8863) -- (axis cs:0,2.8863);
        \draw[dashed] (axis cs:7,0) -- (axis cs:7,3.4459) -- (axis cs:0,3.4459);
    \end{axis}
    \end{tikzpicture}
    \caption{Change in $f(A_i||\bar T^{i+1})$ as we move $x$ weight from $g'$ to $g^*$.}
\end{figure}
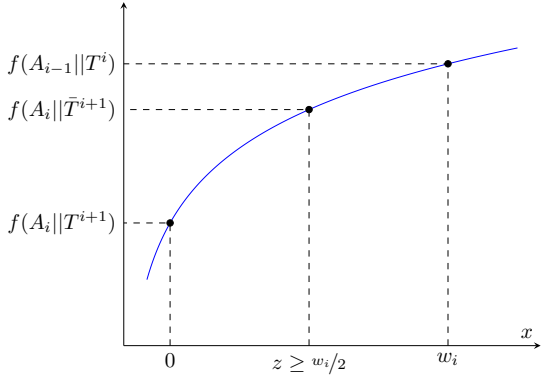

Then, the remaining amount is at most half the uncorrected difference; that is, 
$$ f(A_{i-1} || T^i) - f(A_i || \bar T^{i+1}) \le \frac{1}{2} (f(A_{i-1} || T^i) - f(A_i || T^{i+1})).$$ 
Combining this with the form of inequality (\ref{eq:approx-ineq}) re-established at the beginning of the proof yields 
$ f(A_i || T^i) - f(A_i || \bar T^{i+1}) \le \frac{1}{2} (f(A_i) - f(A_{i-1})),$ which we rearrange to give the desired inequality: 
$$ f(A_i | \bar T^{i+1}) \ge f(A_{i-1} | T^i) - \frac{1}{2} \left(f(A_i) - f(A_{i-1})\right). $$

\textbf{Case 2: $\tau(g') < \frac{1}{2} w_i$.}
Now, $\tau(g')$ is small, so there is not much remaining weight that we can reassign from $g'$ to $g^*$. However, observe that any greedy misstep is due to the fact that the greedy algorithm must choose based only on $\alpha(g')$, with no knowledge of $\tau(g')$. If there is a large $\tau(g')$ that the greedy algorithm does not know about, then the choice to fill $g'$ may have been overly eager, and ultimately ends up being less helpful than expected after the remaining items are assigned.

But here, the fact that $\tau(g')$ is small means that this is \emph{not} the case -- the greedy algorithm was not missing a large piece of information, so the choice based only on $\alpha(g')$ was actually quite good. In particular, it cannot turn out to be much worse than $g^*$, meaning that the difference between $f(A_{i-1} || T^i)$ and $f(A_i || T^{i+1})$ is fairly small. 

In fact, the greedy algorithm's lack of knowledge is most harmful when $\tau(g')$ is large and $\tau(g^*)$ is small. So the \emph{worst} possible outcome for this case occurs when $\tau(g') = \frac{1}{2} w_i$ and $\tau(g^*) = 0$, for which we have
\begin{align*}
    &f(A_{i-1} || T^i) - f(A_i || T^{i+1}) \\
    &\qquad= \sqrt{p(g')} \left(\sqrt{\alpha(g') + w_i/2} - \sqrt{\alpha(g') + 3w_i/2}\right) + \sqrt{p(g^*)} \left(\sqrt{\alpha(g^*) + w_i} - \sqrt{\alpha(g^*)}\right) \\
    &\qquad= \sqrt{p(g')} \left(\sqrt{\alpha(g') + w_i/2} - \sqrt{\alpha(g') + 3w_i/2}\right) + f(A_{i-1} || t_i) - f(A_{i-1}) \\
    &\qquad \le \sqrt{p(g')} \left(\sqrt{\alpha(g') + w_i/2} - \sqrt{\alpha(g') + 3w_i/2}\right) + f(A_i) - f(A_{i-1}).
\end{align*}
Then, this gives
\begin{align*}
    \frac{f(A_{i-1} || T^i) - f(A_i || T^{i+1})}{f(A_i) - f(A_{i-1})} &\le 1 - \frac{\sqrt{p(g')} \left(\sqrt{\alpha(g') + 3w_i/2} - \sqrt{\alpha(g') + w_i/2}\right)}{f(A_i) - f(A_{i-1})} \\
    &= 1 - \frac{\sqrt{\alpha(g') + 3w_i/2} - \sqrt{\alpha(g') + w_i/2}}{\sqrt{\alpha(g') + w_i} - \sqrt{\alpha(g')}}.
\end{align*}
This final expression is minimized when $\alpha(g') = 0$, for which  
$$ \frac{f(A_{i-1} || T^i) - f(A_i || T^{i+1})}{f(A_i) - f(A_{i-1})} \le 1 - \sqrt{2 - \sqrt{3}} \le \frac{1}{2},$$ which rearranges to 
$$f(A_{i}  || T^{i+1}) \ge f(A_{i-1} || T^i) - \frac{1}{2} \left(f(A_i) - f(A_{i-1})\right).$$
Thus, simply taking $\bar T^{i+1} = T^{i+1}$ suffices to give the desired result.
\end{proof}




\section{Calibration in the Distributional Genre Model} \label{sec:distributional}

In this section we consider the general calibrated recommendations problem with a class of distance function between distributions. In this model of \emph{distributional genres}, each item has a specific distribution over genres (as described in Section~\ref{sec:genre_notation}), which we think of as a fine-grained breakdown of all the genres represented by that item.

Note that if we permitted our recommendation list to include repeated elements, then a $(1-1/e)$-approximation algorithm would be possible using a reduction to submodular maximization over a partition matroid constraint (see Appendix~\ref{sec:proofsdistributional} for further details). But realistically, genre mixtures are too specific to have multiple items with identical distributions, and recommendation lists should not show the same item repeatedly. Our main result addresses this setting, providing a $(1-1/e)$-approximation for calibrated recommendation lists \emph{without} repeated elements using SMDR overlap measures.\footnote{One might hope that it would suffice to take a solution with repeats and convert it to a solution without repeats simply by showing items in the order of the first time they appear. Unfortunately, this approach may destroy the submodular structure of the original function, so that the continuous greedy algorithm no longer provides a near-optimal approximation guarantee. Further details are provided in Appendix~\ref{sec:proofsdistributional}.}

To begin, we view a list as an assignment of (at most) one item to each position, so that we consider the ground set of all \emph{item-position pairs} $\{(i,j) \mid  i \in [K], j = \ell\}$ (representing ``item $i$ in position $j$''). Define the laminar family of sets $D_\ell \coloneqq \{(i,j) \mid i \in[K], j \le \ell\}$, and the laminar matroid $\mathcal{M} = (V, \mathcal{I})$, where $R\subset V$ is an independent set in $\mathcal{I}$ if and only if $|R\cap D_\ell| \le \ell$ for all $\ell \in [k]$ (\ie, $R$ assigns at most $\ell$ items to the first $\ell$ positions, essentially corresponding to a ``valid'' list).

We now observe that there is a correspondence between recommendation lists and laminar matroid bases: any list assigns exactly $\ell$ items to the first $\ell$ slots for all $\ell \in [k]$ (and is thus a basis); any basis can also be converted into a list solely by promoting items upwards (and by the strong monotonicity property, this transformation preserves the value of the calibration objective). Then, it suffices to optimize over matroid bases using the continuous greedy algorithm and pipage rounding algorithm technique of \cite{calinescu2011}, then convert the approximately-optimal basis back to an approximately-optimal list. 

\begin{restatable}{proposition}{settoseq}
\label{prop:settoseq}
Given a basis $R\in \mathcal{I}$, we can construct a length-$k$ list $\pi$ such that $G(\pi) \ge F_G(R)$.
\end{restatable}

\begin{proof}
For every item $i$, we define $\ell_R (i)$ to be the first position that $i$ occurs in $R$; that is, $\ell_R (i)\coloneqq \min \left\{ j\in [k] |  (i,j) \in R \right\}$, or $\ell_R(i) = k+1$ if no such $j$ exists. We also introduce the notation of $w_{k+1}=0$. Sort the items in increasing order of $\ell_R(\cdot)$ (breaking ties arbitrarily), and call this sequence $\pi$. We claim that $G(\pi) \ge F_G(R)$.

Consider an arbitrary item $j$. By definition of the laminar matroid, 
$$|R \cap D_{\ell_R(i)}| = \sum_{y = 1}^k \sum_{x=1}^{\ell_R (i)} \1{(x,y) \in R} \le \ell_R(i). $$ 
The summation is an upper bound on the number of items $x$ with $(x,y) \in R$ for some $y \le \ell_R (i)$. But these are exactly the items with $\ell_R(x) \le \ell_R(i)$ (including $i$ itself), and therefore the items that can appear before $i$ in $\pi$. So the position at which $i$ appears in $\pi$, denoted $\pi^{-1}(i)$, is less than or equal to $\ell_R(i)$. This implies $w_{\pi^{-1}(i)} \ge w_{\ell_R (i)}$ for all $i$. 

Now, observe that $R^{\le j} \setminus R^{\le j-1}$ is exactly the set of items which appear for the \emph{first} time in position $j$; thus $\ell_R(i) = j$ for all $i \in R^{\le j} \setminus R^{\le j-1}$. Additionally, $R^{\le 1} \subseteq R^{\le 2} \subseteq \dots \subseteq R^{\le k} \subseteq [K]$. Then, for any genre $g$, we have
\begin{align*}
\sum_{j=1}^{k} w_j \left(\sum_{i\in R^{\le j} \setminus R^{\le j-1}} q_i(g) \right) = \sum_{j=1}^k \sum_{i \in R^{\le j} \setminus R^{\le j-1}} w_{\ell_R(i)} q_{i}(g) &= \sum_{i \in R^{\le k}} w_{\ell_R(i)} q_{i}(g) \\ 
&\le \sum_{i \in [K]} w_{\pi^{-1}(i)} q_{i}(g) \\
&= \sum_{j=1}^k  w_j q_{\pi(j)} (g).
\end{align*}

Then, since $G$ is non-decreasing with respect to all $q(g)$, we have $G(R) \le F_G(\pi)$. 
\end{proof}

\begin{restatable}{proposition}{maxGf}
\label{prop:maxGf}
$\max_{R \in \mathcal{I}} F_G(R) \ge \max G(\pi)$.
\end{restatable}
\begin{proof}
Let $\pi^* \coloneqq \argmax_{\pi} G(\pi)$, and define $R^* \coloneqq \left\{ {((\pi^*)^{-1}(j),j)} \mid j \in [k] \right\}$ as the set corresponding to the item-position pairs in $\pi^*$.

By construction, $F_G(R^*) = G(\pi^*)$, and $R^* \in \mathcal{I}$. Then, by monotonicity the the maximum value over independent sets is attained by a basis, and we get
$$\max_{R \in \mathcal{I}} F_G(R)  \ge F_G(R^*) = G(\pi^*) = \max G(\pi).$$ 
\end{proof}

\begin{restatable}{theorem}{partitionthm}
\label{thm:sqrt1e}
There exists a $(1-1/e)$-approximation algorithm for the calibration problem with distributional genres using any SMDR overlap measure $G$.
\end{restatable}
\begin{proof}
Since $G$ is an SMDR overlap measure, $F_G$ is a monotone submodular function. Then, the continuous greedy algorithm and pipage rounding technique of \cite{calinescu2011} finds an independent set $\bar{R} \in \mathcal{I}$ such that $F_G(\bar{R}) \ge (1-1/e) \max_{R \in \mathcal{I}} F_G(R)$. We can assume $\bar R$ is a basis. By Proposition~\ref{prop:maxGf}, $F_G(\bar{R}) \ge (1-1/e)\max G(\pi)$.

Using Proposition~\ref{prop:settoseq}, we can convert $\bar{R}$ into a sequence $\bar{\pi}$ such that $G(\bar{\pi}) \ge F_G(\bar{R})$. Now observe that $G(\bar{\pi}) \ge (1-1/e)\max G(\pi)$, so we take $\bar{\pi}$ to be the output of the algorithm.
\end{proof}



\section{Conclusion} \label{sec:conclusion}

In this paper, we have studied the problem of calibrating a recommendation list to match a user's interests, where user attention decays over the course of the list. We have introduced the notion of overlap measures, as a generalization of the measures used to quantify calibration under two different models of genre distributions. In the first model, where every item belongs to a single \emph{discrete genre}, by defining a property we call ordered submodularity and utilizing a careful bin-packing argument, we have shown that the greedy algorithm is a $2/3$-approximation. In the second model of \emph{distributional genres}, where each item has a fine-grained mixture of genre percentages, we have extended tools from constrained submodular optimization to supply a $(1-1/e)$-approximation algorithm. Prior work had highlighted the importance of the order of items due to attention decay but had left open the question of provable guarantees for calibration on these types of sequences; this prior work obtained guarantees only under the assumption that the ordering of items does not matter. Now, our work has provided the first performance guarantees for near-optimal calibration of recommendation \emph{lists}, working within the models
of user attention that form the underpinnings of applications in search and recommendation.

Finally, we highlight a number of directions for further work suggested by our results. First, it is interesting to consider the greedy algorithm for the calibration problem with discrete genres and ask whether the approximation bound of $2/3$ is tight, or if it can be sharpened using an alternative analysis technique. Additionally, we ask whether $(1-1/e)$ and $2/3$ are the best possible approximation guarantees possible for the distributional and discrete genre models, respectively, or if there exists a polynomial time approximation algorithm that achieves a stronger constant factor under either model. As noted earlier, both models of calibration with decaying attention are amenable to the general framework of submodular optimization, but these tools are limited to an approximation guarantee of $(1-1/e)$. In the discrete genre model, by using different techniques we surpass this barrier and obtain a stronger guarantee; might the same be possible in the distributional genre model?

To further investigate the performance of our algorithms, it may be useful to parametrize worst-case instances of the calibration problem, since we found through basic computational simulations that the greedy solution tends to be very close to optimal across many randomly generated instances (see Appendix~\ref{sec:comp_experiments} for details). Another potential direction is constructing additional families of overlap measures, or deriving a broader characterization of functional forms that satisfy the MDR and SMDR properties so that they may be used with our algorithms. As personalized recommendations become increasingly commonplace and explicitly optimized, the answers to these questions will be essential in developing tools to better understand the interplay between relevance, calibration, and other notions of diversity in these systems. 

\bibliographystyle{ACM-Reference-Format}
\bibliography{main.bib}


\begin{thebibliography}{24}


\ifx \showCODEN    \undefined \def \showCODEN     #1{\unskip}     \fi
\ifx \showDOI      \undefined \def \showDOI       #1{#1}\fi
\ifx \showISBNx    \undefined \def \showISBNx     #1{\unskip}     \fi
\ifx \showISBNxiii \undefined \def \showISBNxiii  #1{\unskip}     \fi
\ifx \showISSN     \undefined \def \showISSN      #1{\unskip}     \fi
\ifx \showLCCN     \undefined \def \showLCCN      #1{\unskip}     \fi
\ifx \shownote     \undefined \def \shownote      #1{#1}          \fi
\ifx \showarticletitle \undefined \def \showarticletitle #1{#1}   \fi
\ifx \showURL      \undefined \def \showURL       {\relax}        \fi
\providecommand\bibfield[2]{#2}
\providecommand\bibinfo[2]{#2}
\providecommand\natexlab[1]{#1}
\providecommand\showeprint[2][]{arXiv:#2}

\bibitem[Abdollahpouri et~al\mbox{.}(2020)]%
        {calAMBM20}
\bibfield{author}{\bibinfo{person}{Himan Abdollahpouri},
  \bibinfo{person}{Masoud Mansoury}, \bibinfo{person}{Robin Burke}, {and}
  \bibinfo{person}{Bamshad Mobasher}.} \bibinfo{year}{2020}\natexlab{}.
\newblock \showarticletitle{The Connection Between Popularity Bias,
  Calibration, and Fairness in Recommendation}. In
  \bibinfo{booktitle}{\emph{Proceedings of the 14th ACM Conference on
  Recommender Systems}} (Virtual Event, Brazil) \emph{(\bibinfo{series}{RecSys
  '20})}. \bibinfo{publisher}{Association for Computing Machinery},
  \bibinfo{address}{New York, NY, USA}, \bibinfo{pages}{726–731}.
\newblock
\showISBNx{9781450375832}
\urldef\tempurl%
\url{https://doi.org/10.1145/3383313.3418487}
\showDOI{\tempurl}


\bibitem[Agrawal et~al\mbox{.}(2009)]%
        {agrawal2009}
\bibfield{author}{\bibinfo{person}{Rakesh Agrawal}, \bibinfo{person}{Sreenivas
  Gollapudi}, \bibinfo{person}{Alan Halverson}, {and} \bibinfo{person}{Samuel
  Ieong}.} \bibinfo{year}{2009}\natexlab{}.
\newblock \showarticletitle{Diversifying search results}. In
  \bibinfo{booktitle}{\emph{Proceedings of the second ACM international
  conference on web search and data mining}}. \bibinfo{pages}{5--14}.
\newblock
\urldef\tempurl%
\url{https://dl.acm.org/doi/10.1145/1498759.1498766}
\showURL{%
\tempurl}


\bibitem[Alaei et~al\mbox{.}(2019)]%
        {Alaei19}
\bibfield{author}{\bibinfo{person}{Saeed Alaei}, \bibinfo{person}{Ali
  Makhdoumi}, {and} \bibinfo{person}{Azarakhsh Malekian}.}
  \bibinfo{year}{2019}\natexlab{}.
\newblock \bibinfo{title}{Maximizing Sequence-Submodular Functions and its
  Application to Online Advertising}.
\newblock
\newblock
\showeprint[arxiv]{1009.4153}~[cs.DM]


\bibitem[Asadpour et~al\mbox{.}(2022)]%
        {ec22seqsub}
\bibfield{author}{\bibinfo{person}{Arash Asadpour}, \bibinfo{person}{Rad
  Niazadeh}, \bibinfo{person}{Amin Saberi}, {and} \bibinfo{person}{Ali
  Shameli}.} \bibinfo{year}{2022}\natexlab{}.
\newblock \showarticletitle{Sequential Submodular Maximization and Applications
  to Ranking an Assortment of Products}. In
  \bibinfo{booktitle}{\emph{Proceedings of the 23rd ACM Conference on Economics
  and Computation}} (Boulder, CO, USA) \emph{(\bibinfo{series}{EC '22})}.
  \bibinfo{publisher}{Association for Computing Machinery},
  \bibinfo{address}{New York, NY, USA}, \bibinfo{pages}{817}.
\newblock
\showISBNx{9781450391504}
\urldef\tempurl%
\url{https://doi.org/10.1145/3490486.3538361}
\showDOI{\tempurl}


\bibitem[Ashkan et~al\mbox{.}(2015)]%
        {ashkan2015}
\bibfield{author}{\bibinfo{person}{Azin Ashkan}, \bibinfo{person}{Branislav
  Kveton}, \bibinfo{person}{Shlomo Berkovsky}, {and} \bibinfo{person}{Zheng
  Wen}.} \bibinfo{year}{2015}\natexlab{}.
\newblock \showarticletitle{Optimal greedy diversity for recommendation}. In
  \bibinfo{booktitle}{\emph{Twenty-Fourth International Joint Conference on
  Artificial Intelligence}}.
\newblock


\bibitem[Bernardini et~al\mbox{.}(2021)]%
        {Bernardini21}
\bibfield{author}{\bibinfo{person}{Sara Bernardini}, \bibinfo{person}{Fabio
  Fagnani}, {and} \bibinfo{person}{Chiara Piacentini}.}
  \bibinfo{year}{2021}\natexlab{}.
\newblock \showarticletitle{A unifying look at sequence submodularity}.
\newblock \bibinfo{journal}{\emph{Artificial Intelligence}}
  \bibinfo{volume}{297} (\bibinfo{year}{2021}), \bibinfo{pages}{103486}.
\newblock
\showISSN{0004-3702}
\urldef\tempurl%
\url{https://doi.org/10.1016/j.artint.2021.103486}
\showDOI{\tempurl}


\bibitem[Calinescu et~al\mbox{.}(2011)]%
        {calinescu2011}
\bibfield{author}{\bibinfo{person}{Gruia Calinescu}, \bibinfo{person}{Chandra
  Chekuri}, \bibinfo{person}{Martin Pal}, {and} \bibinfo{person}{Jan
  Vondr{\'a}k}.} \bibinfo{year}{2011}\natexlab{}.
\newblock \showarticletitle{Maximizing a monotone submodular function subject
  to a matroid constraint}.
\newblock \bibinfo{journal}{\emph{SIAM J. Comput.}} \bibinfo{volume}{40},
  \bibinfo{number}{6} (\bibinfo{year}{2011}), \bibinfo{pages}{1740--1766}.
\newblock


\bibitem[da~Silva and Durão(2022)]%
        {calCD22}
\bibfield{author}{\bibinfo{person}{Diego~Corrêa da Silva} {and}
  \bibinfo{person}{Frederico~Araújo Durão}.} \bibinfo{year}{2022}\natexlab{}.
\newblock \bibinfo{title}{Introducing a Framework and a Decision Protocol to
  Calibrate Recommender Systems}.
\newblock
\newblock
\urldef\tempurl%
\url{https://doi.org/10.48550/ARXIV.2204.03706}
\showDOI{\tempurl}


\bibitem[Fessenden(2018)]%
        {NNGfold}
\bibfield{author}{\bibinfo{person}{Therese Fessenden}.}
  \bibinfo{year}{2018}\natexlab{}.
\newblock \bibinfo{title}{Scrolling and Attention}.
\newblock
  \bibinfo{howpublished}{\url{https://www.nngroup.com/articles/scrolling-and-attention/}}.
\newblock
\newblock
\shownote{Accessed: 2022-02-04}.


\bibitem[Kowald et~al\mbox{.}(2023)]%
        {kowald2023study}
\bibfield{author}{\bibinfo{person}{Dominik Kowald}, \bibinfo{person}{Gregor
  Mayr}, \bibinfo{person}{Markus Schedl}, {and} \bibinfo{person}{Elisabeth
  Lex}.} \bibinfo{year}{2023}\natexlab{}.
\newblock \bibinfo{title}{A Study on Accuracy, Miscalibration, and Popularity
  Bias in Recommendations}.
\newblock
\newblock
\showeprint[arxiv]{2303.00400}~[cs.IR]


\bibitem[Liao et~al\mbox{.}(2022)]%
        {userstudy}
\bibfield{author}{\bibinfo{person}{Mengqi Liao}, \bibinfo{person}{S.~Shyam
  Sundar}, {and} \bibinfo{person}{Joseph B.~Walther}.}
  \bibinfo{year}{2022}\natexlab{}.
\newblock \showarticletitle{User Trust in Recommendation Systems: A Comparison
  of Content-Based, Collaborative and Demographic Filtering}. In
  \bibinfo{booktitle}{\emph{Proceedings of the 2022 CHI Conference on Human
  Factors in Computing Systems}} (New Orleans, LA, USA)
  \emph{(\bibinfo{series}{CHI '22})}. \bibinfo{publisher}{Association for
  Computing Machinery}, \bibinfo{address}{New York, NY, USA}, Article
  \bibinfo{articleno}{486}, \bibinfo{numpages}{14}~pages.
\newblock
\showISBNx{9781450391573}
\urldef\tempurl%
\url{https://doi.org/10.1145/3491102.3501936}
\showDOI{\tempurl}


\bibitem[Lin et~al\mbox{.}(2020)]%
        {cal2}
\bibfield{author}{\bibinfo{person}{Kun Lin}, \bibinfo{person}{Nasim Sonboli},
  \bibinfo{person}{Bamshad Mobasher}, {and} \bibinfo{person}{Robin Burke}.}
  \bibinfo{year}{2020}\natexlab{}.
\newblock \showarticletitle{Calibration in Collaborative Filtering Recommender
  Systems: A User-Centered Analysis}. In \bibinfo{booktitle}{\emph{Proceedings
  of the 31st ACM Conference on Hypertext and Social Media}} (Virtual Event,
  USA) \emph{(\bibinfo{series}{HT '20})}. \bibinfo{publisher}{Association for
  Computing Machinery}, \bibinfo{address}{New York, NY, USA},
  \bibinfo{pages}{197–206}.
\newblock
\showISBNx{9781450370981}
\urldef\tempurl%
\url{https://doi.org/10.1145/3372923.3404793}
\showDOI{\tempurl}


\bibitem[McNee et~al\mbox{.}(2006)]%
        {McNee2006}
\bibfield{author}{\bibinfo{person}{Sean~M. McNee}, \bibinfo{person}{John
  Riedl}, {and} \bibinfo{person}{Joseph~A. Konstan}.}
  \bibinfo{year}{2006}\natexlab{}.
\newblock \showarticletitle{Being Accurate is Not Enough: How Accuracy Metrics
  Have Hurt Recommender Systems}. In \bibinfo{booktitle}{\emph{CHI '06 Extended
  Abstracts on Human Factors in Computing Systems}} (Montr\'{e}al, Qu\'{e}bec,
  Canada) \emph{(\bibinfo{series}{CHI EA '06})}.
  \bibinfo{publisher}{Association for Computing Machinery},
  \bibinfo{address}{New York, NY, USA}, \bibinfo{pages}{1097–1101}.
\newblock
\showISBNx{1595932984}
\urldef\tempurl%
\url{https://doi.org/10.1145/1125451.1125659}
\showDOI{\tempurl}


\bibitem[Mitrovic et~al\mbox{.}(2018)]%
        {Mitrovic18}
\bibfield{author}{\bibinfo{person}{Marko Mitrovic}, \bibinfo{person}{Moran
  Feldman}, \bibinfo{person}{Andreas Krause}, {and} \bibinfo{person}{Amin
  Karbasi}.} \bibinfo{year}{2018}\natexlab{}.
\newblock \showarticletitle{Submodularity on Hypergraphs: From Sets to
  Sequences}. In \bibinfo{booktitle}{\emph{Proceedings of the Twenty-First
  International Conference on Artificial Intelligence and Statistics}}
  \emph{(\bibinfo{series}{Proceedings of Machine Learning Research},
  Vol.~\bibinfo{volume}{84})}, \bibfield{editor}{\bibinfo{person}{Amos Storkey}
  {and} \bibinfo{person}{Fernando Perez-Cruz}} (Eds.).
  \bibinfo{publisher}{PMLR}, \bibinfo{pages}{1177--1184}.
\newblock
\urldef\tempurl%
\url{https://proceedings.mlr.press/v84/mitrovic18a.html}
\showURL{%
\tempurl}


\bibitem[Naghiaei et~al\mbox{.}(2022)]%
        {calNRAS22}
\bibfield{author}{\bibinfo{person}{Mohammadmehdi Naghiaei},
  \bibinfo{person}{Hossein~A. Rahmani}, \bibinfo{person}{Mohammad Aliannejadi},
  {and} \bibinfo{person}{Nasim Sonboli}.} \bibinfo{year}{2022}\natexlab{}.
\newblock \bibinfo{title}{Towards Confidence-aware Calibrated Recommendation}.
\newblock
\newblock
\urldef\tempurl%
\url{https://doi.org/10.48550/ARXIV.2208.10192}
\showDOI{\tempurl}


\bibitem[Pan et~al\mbox{.}(2007)]%
        {pan2007google}
\bibfield{author}{\bibinfo{person}{Bing Pan}, \bibinfo{person}{Helene
  Hembrooke}, \bibinfo{person}{Thorsten Joachims}, \bibinfo{person}{Lori
  Lorigo}, \bibinfo{person}{Geri Gay}, {and} \bibinfo{person}{Laura Granka}.}
  \bibinfo{year}{2007}\natexlab{}.
\newblock \showarticletitle{In Google we trust: Users’ decisions on rank,
  position, and relevance}.
\newblock \bibinfo{journal}{\emph{Journal of computer-mediated communication}}
  \bibinfo{volume}{12}, \bibinfo{number}{3} (\bibinfo{year}{2007}),
  \bibinfo{pages}{801--823}.
\newblock


\bibitem[Seymen et~al\mbox{.}(2021)]%
        {calSAM21}
\bibfield{author}{\bibinfo{person}{Sinan Seymen}, \bibinfo{person}{Himan
  Abdollahpouri}, {and} \bibinfo{person}{Edward~C. Malthouse}.}
  \bibinfo{year}{2021}\natexlab{}.
\newblock \showarticletitle{A Constrained Optimization Approach for Calibrated
  Recommendations}. In \bibinfo{booktitle}{\emph{Proceedings of the 15th ACM
  Conference on Recommender Systems}} (Amsterdam, Netherlands)
  \emph{(\bibinfo{series}{RecSys '21})}. \bibinfo{publisher}{Association for
  Computing Machinery}, \bibinfo{address}{New York, NY, USA},
  \bibinfo{pages}{607–612}.
\newblock
\showISBNx{9781450384582}
\urldef\tempurl%
\url{https://doi.org/10.1145/3460231.3478857}
\showDOI{\tempurl}


\bibitem[Steck(2018)]%
        {Steck18}
\bibfield{author}{\bibinfo{person}{Harald Steck}.}
  \bibinfo{year}{2018}\natexlab{}.
\newblock \showarticletitle{Calibrated Recommendations}. In
  \bibinfo{booktitle}{\emph{Proceedings of the 12th ACM Conference on
  Recommender Systems}} (Vancouver, British Columbia, Canada)
  \emph{(\bibinfo{series}{RecSys '18})}. \bibinfo{publisher}{Association for
  Computing Machinery}, \bibinfo{address}{New York, NY, USA},
  \bibinfo{pages}{154–162}.
\newblock
\showISBNx{9781450359016}
\urldef\tempurl%
\url{https://doi.org/10.1145/3240323.3240372}
\showDOI{\tempurl}


\bibitem[Streeter and Golovin(2008)]%
        {StreeterGolovin}
\bibfield{author}{\bibinfo{person}{Matthew Streeter} {and}
  \bibinfo{person}{Daniel Golovin}.} \bibinfo{year}{2008}\natexlab{}.
\newblock \showarticletitle{An online algorithm for maximizing submodular
  functions}.
\newblock \bibinfo{journal}{\emph{Advances in Neural Information Processing
  Systems}}  \bibinfo{volume}{21} (\bibinfo{year}{2008}).
\newblock


\bibitem[Tschiatschek et~al\mbox{.}(2017)]%
        {Tschiatschek2017}
\bibfield{author}{\bibinfo{person}{Sebastian Tschiatschek},
  \bibinfo{person}{Adish Singla}, {and} \bibinfo{person}{Andreas Krause}.}
  \bibinfo{year}{2017}\natexlab{}.
\newblock \showarticletitle{Selecting sequences of items via submodular
  maximization}. In \bibinfo{booktitle}{\emph{Thirty-First AAAI Conference on
  Artificial Intelligence}}.
\newblock


\bibitem[Udwani(2021)]%
        {udwani21order}
\bibfield{author}{\bibinfo{person}{Rajan Udwani}.}
  \bibinfo{year}{2021}\natexlab{}.
\newblock \bibinfo{title}{Submodular Order Functions and Assortment
  Optimization}.
\newblock
\newblock
\urldef\tempurl%
\url{https://doi.org/10.48550/ARXIV.2107.02743}
\showDOI{\tempurl}


\bibitem[Williams(2012)]%
        {williamsZipf}
\bibfield{author}{\bibinfo{person}{Hugh~E. Williams}.}
  \bibinfo{year}{2012}\natexlab{}.
\newblock \bibinfo{title}{Clicks in search}.
\newblock
  \bibinfo{howpublished}{\url{https://hughewilliams.com/2012/04/12/clicks-in-search/}}.
\newblock
\newblock
\shownote{Accessed: 2022-02-04}.


\bibitem[Zhang et~al\mbox{.}(2022)]%
        {zhang2022ranking}
\bibfield{author}{\bibinfo{person}{Guangyi Zhang}, \bibinfo{person}{Nikolaj
  Tatti}, {and} \bibinfo{person}{Aristides Gionis}.}
  \bibinfo{year}{2022}\natexlab{}.
\newblock \showarticletitle{Ranking with submodular functions on a budget}.
\newblock \bibinfo{journal}{\emph{Data mining and knowledge discovery}}
  \bibinfo{volume}{36}, \bibinfo{number}{3} (\bibinfo{year}{2022}),
  \bibinfo{pages}{1197--1218}.
\newblock


\bibitem[Zhang et~al\mbox{.}(2016)]%
        {Zhang16}
\bibfield{author}{\bibinfo{person}{Zhenliang Zhang}, \bibinfo{person}{Edwin
  K.~P. Chong}, \bibinfo{person}{Ali Pezeshki}, {and} \bibinfo{person}{William
  Moran}.} \bibinfo{year}{2016}\natexlab{}.
\newblock \showarticletitle{String Submodular Functions With Curvature
  Constraints}.
\newblock \bibinfo{journal}{\emph{IEEE Trans. Automat. Control}}
  \bibinfo{volume}{61}, \bibinfo{number}{3} (\bibinfo{year}{2016}),
  \bibinfo{pages}{601--616}.
\newblock
\urldef\tempurl%
\url{https://doi.org/10.1109/TAC.2015.2440566}
\showDOI{\tempurl}


\end{thebibliography}

\newpage
\appendix

\section{Survey of existing submodularity frameworks} \label{sec:extended_survey}
\cite{Alaei19} and \cite{Zhang16} introduce \emph{sequence-submodularity} and \emph{string submodularity}, with $(1-1/e)$-approximate greedy algorithms matching traditional submodularity. Both definitions require extremely strong mononicity conditions including \emph{postfix monotonicity}, which states that for any sequences $A$ and $B$ and their concatenation $A||B$, it must hold that $f(A||B) \ge f(B)$ \cite{StreeterGolovin}. But this frequently is not a natural property to assume (for instance, prepending a ``bad'' movie to the front of a list and forcing ``good'' movies to move downwards will not improve calibration).

\cite{Tschiatschek2017} and \cite{Mitrovic18} study submodularity in sequences using graphs and hypergraphs, respectively. However, this approach only models simple sequential dependencies between individual items, not more complex phenomena such as attention decay.

\cite{Bernardini21} propose a framework in which the set of all elements has a total ordering according to some property $g$. Denoting the subsequence of the first $i$ elements in the list as $S_i$, they consider functions of the form $f(s_1 \dots s_k) = \sum_{i=1}^k g(s_i) \cdot [ h(S_i) - h(S_{i-1}) ]$ for any function $g$ and any monotone submodular set function $h$. The marginal increase due to each element is weighted solely based on its \emph{identity} and not its \emph{rank} (in contrast to the rank-based weights of \cite{Steck18}); while a valid assumption in some applications, this does not hold for sequential attention decay.

\cite{ec22seqsub} study the maximization of \emph{sequential submodular} functions of the specific form $f(S) = \sum_{i=1}^k g_i \cdot h_i (S_i)$, which also include the rank-based weighted marginal increase functions referred to above, and provide a $(1-1/e)$-approximation via a matroid reduction and continuous greedy algorithm. \cite{zhang2022ranking} formulate the \emph{max-submodular ranking} problem, a generalization which incorporates budget constraints on the lengths of the prefixes $S_i$, and study versions of the greedy algorithm under varying conditions. However, these techniques do not extend readily to functions that cannot be expressed using sums of increasing nested subsequences $S_i$.

Lastly, we note that \cite{udwani21order} develops the similarly named class of \emph{submodular order functions}. However, these are \emph{set functions} that display a limited form of submodularity only with respect to a certain permutation (or ordering) of the ground set. As such, this property is designed for optimization over sets, not lists, and does not model the calibration problem at hand.


\section{Deferred proofs}\label{sec:proofs}
Here, we provide complete proofs deferred from Sections~\ref{sec:calibration_overlap}, \ref{sec:distributional}, and \ref{sec:discrete} in the main text. 

\subsection{Proofs and discussion from Section~\ref{sec:overlap_families}} \label{sec:proofsMDR}

\fdivergence*
\begin{proof}[Proof of monotonicity]
For any set $R$ and each item $j$, we define $\ell_R (i)\coloneqq \min \left\{ j\in [k] |  (i,j) \in R \right\}$ (or $\ell_R(i) = k+1$ if no such $j$ exists) as the earliest position in which $R$ places item $i$. We also define $w_{k+1} = 0$. 

 Observe that $R^{\le j}  \setminus R^{\le j-1}$ is exactly the set of items which appear for the \emph{first} time in position $j$; thus $\ell_R(i) = j$ for all $i \in R^{\le j} \setminus R^{\le j-1}$. Additionally, $R^{\le 1} \subseteq R^{\le 2} \subseteq \dots \subseteq R^{\le k}$.

Consider a superset $T \supseteq R$. It is clear from the definition that $\ell_T(i) \le \ell_R(i)$ for all $i$, so $w_{\ell_T(i)} \ge w_{\ell_R(i)}$. Further, $R^{\le k} \subseteq T^{\le k}$. Then, for any $x$, we have
\begin{align*}
\sum_{j=1}^{k} w_j \left(\sum_{i\in R^{\le j} \setminus R^{\le j-1}} q_i(x) \right) &= \sum_{j=1}^k \sum_{i \in R^{\le j} \setminus R^{\le j-1}} w_{\ell_R(i)} q_{i}(x) \\
&= \sum_{i \in R^{\le k}} w_{\ell_R(i)} q_{i}(x) \\ 
&\le \sum_{i \in T^{\le k}} w_{\ell_T(i)} q_{i}(x) \\
&= \sum_{j=1}^k w_j \left(\sum_{i \in T^{\le j} \setminus T^{\le j-1}} q_{i}(x) \right),
\end{align*}
so $q^{(R)}$ is coordinate-wise dominated by $q^{(T)}$. Since $D_f$ is monotone over subdistributions, we have $D_f(p, q^{(R)}) \ge D_f(p, q^{(T)}) \implies F_G(R) \le F_G(T)$. 
\end{proof}

\begin{proof}[Proof of submodularity] 
We now show that for any sets $R\subseteq T$ and element $(a,b)$, $F_G(R \cup \{(a,b)\}) - F_G(R) \ge F_G(T\cup \{(a,b)\}) - F_G(T)$. If $(a,b) \in T$, the inequality clearly holds since $F_G$ is monotone, so suppose $(a,b) \notin T$. To simplify notation, we also write $R' = R\cup \{(a,b)\}$, $T' = T\cup \{(a,b)\}$. Observe that for all $i\ne a$, $\ell_{R'}(i) = \ell_R(i) \ge \ell_T(i) = \ell_{T'}(i)$. We also have $\ell_{R'}(a) \le \ell_R(a)$ and $\ell_{T'}(a) \le \ell_T(a)$. 

We claim that for all $x$, $q^{(R')}(x) - q^{(R)}(x) \ge q^{(T')}(x) - q^{(T)}(x)$. We take three cases: 

\paragraph*{Case 1: $a \in R^{\le k} (\subseteq T^{\le k})$.}

Then $(R')^{\le k} = R^{\le k}$ and $(T')^{\le k} = T^{\le k}$. We take subcases based on $\ell_{T}(a)$, and show that in both subcases, $w_{\ell_{R'}(a)} - w_{\ell_R(a)} \ge w_{\ell_{T'}(a)} - w_{\ell_T(a)}$.
\begin{itemize}
    \item $\ell_T(a) \le b$: Then $\ell_{T'}(a) = \ell_T(a)$, while $\ell_{R'}(a) \le \ell_R(a)$. So $w_{\ell_{R'}(a)} - w_{\ell_{R}(a)} \ge 0 = w_{\ell_{T'}(a)} - w_{\ell_{T}(a)}$.
    
    \item $b< \ell_T(a)$ ($\le \ell_R(a)$): Then $\ell_{R'}(a) = \ell_{T'}(a) = b$, so $w_{\ell_{R'}(a)} - w_{\ell_R(a)} = w_b -  w_{\ell_R(a)} \ge w_b -  w_{\ell_T(a)} \ge w_{\ell_{T'}(a)} - w_{\ell_T(a)}$.
\end{itemize}
Then, we have
\begin{align*}
q^{(R')}(x) - q^{(R)}(x) &= \sum_{i \in (R')^{\le k}} w_{\ell_{R'} (i)} q_i(x) - \sum_{i \in R^{\le k}} w_{\ell_R (i)} q_i(x) \\
&= \sum_{i \in R^{\le k}} w_{\ell_{R'} (i)} q_i(x) - \sum_{i \in R^{\le k}} w_{\ell_R (i)} q_i(x) \\
&= (w_{\ell_{R'} (a)} - w_{\ell_{R} (a)}) q_a(x) \\
&\ge (w_{\ell_{T'} (a)} - w_{\ell_{T} (a)}) q_a(x) \\
&= \sum_{i \in (T')^{\le k}} w_{\ell_{T'} (i)} q_i(x) - \sum_{i \in T^{\le k}} w_{\ell_T (i)} q_i(x) \\
&= q^{(T')}(x) - q^{(T)}(x).
\end{align*}

\paragraph*{Case 2: $a\notin R^{\le k}$, $a\in T^{\le k}$.}

In this case, $(R')^{\le k} = R^{\le k} \cup \{a\}$ and $(T')^{\le k} = T^{\le k}$. We have $\ell_{R'}(a) = b$. If $\ell_T(a) < b$, then $\ell_{T'}(a) = \ell_T(a)$, so $w_{\ell_{T'}(a)} - w_{\ell_T(a)} = 0$. If $\ell_T(a) \ge b$, then $\ell_{T'}(a) = b$, so $w_{\ell_{T'}(a)} - w_{\ell_T(a)} = w_b - w_{\ell_T(a)}$. In either case, we have $w_{\ell_{T'}(a)} - w_{\ell_T(a)} \le w_b$, so we have 
\begin{align*}
q^{(R')}(x) - q^{(R)}(x) &= \sum_{i \in (R')^{\le k}} w_{\ell_{R'} (i)} q_i(x) - \sum_{i \in R^{\le k}} w_{\ell_R (i)} q_i(x)\\
&= \left(w_b q_a(x) + \sum_{i \in R^{\le k}} w_{\ell_{R} (i)} q_i(x)\right) \\
    &\quad - \sum_{i \in R^{\le k}} w_{\ell_R (i)} q_i(x) \\
&= w_b q_a(x) \\
&\ge (w_{\ell_{T'}(a)} - w_{\ell_T(a)}) q_a(x) \\
& =\sum_{i \in (T')^{\le k}} w_{\ell_{T'} (i)} q_i(x) - \sum_{i \in T^{\le k}} w_{\ell_T (i)} q_i(x) \\
&= q^{(T')}(x) - q^{(T)}(x).
\end{align*}

\paragraph*{Case 3: $a \notin T^{\le k} (\supseteq R^{\le k})$.}

In this case, $(R')^{\le k} = R^{\le k} \cup \{a\}$ and $(T')^{\le k} = T^{\le k} \cup\{a\}$, and $\ell_{R'}(a) = \ell_{T'}(a) = b$. Then,
\begin{align*}
q^{(R')}(x) - q^{(R)}(x) &= \sum_{i \in (R')^{\le k}} w_{\ell_{R'} (i)} q_i(x) - \sum_{i \in R^{\le k}} w_{\ell_R (i)} q_i(x) \\
&= \left(w_b q_a(x) + \sum_{i \in R^{\le k}} w_{\ell_{R} (i)} q_i(x)\right) \\
    &\quad - \sum_{i \in R^{\le k}} w_{\ell_R (i)} q_i(x) \\
&= w_b q_a(x) \\
&=\sum_{i \in (T')^{\le k}} w_{\ell_{T'} (i)} q_i(x) - \sum_{i \in T^{\le k}} w_{\ell_T (i)} q_i(x) \\
&= q^{(T')}(x) - q^{(T)}(x).
\end{align*}

Lastly, we can compute
\begin{align*}
\frac{\partial G_{D_f}}{\partial q(x)} &= -\frac{\partial}{\partial q(x)} \left( f\left(\frac{p(x)}{q(x)}\right) q(x) \right) \\
&= f'\left(\frac{p(x)}{q(x)}\right) \cdot \frac{p(x)}{q(x)^2} \cdot q(x) - f\left(\frac{p(x)}{q(x)}\right)\cdot 1 \\
&= f'\left(\frac{p(x)}{q(x)}\right) \cdot \frac{p(x)}{q(x)} - f\left(\frac{p(x)}{q(x)}\right),\\
\frac{\partial^2 G_{D_f}}{\partial q(x)^2 } &= -f''\left(\frac{p(x)}{q(x)}\right) \cdot \frac{p(x)}{q(x)^2} \cdot \frac{p(x)}{q(x)} - f'\left(\frac{p(x)}{q(x)}\right) \cdot \frac{p(x)}{q(x)^2} \\
    &\quad+ f'\left(\frac{p(x)}{q(x)}\right) \cdot \frac{p(x)}{q(x)^2} \\
&= -f''\left(\frac{p(x)}{q(x)}\right) \cdot \frac{p(x)^2}{q(x)^3} \\
&\le 0,
\end{align*}
since $p(x), q(x), f'' \ge 0$ (by convexity of $f$), so $G_{D_f}$ is concave in every $q(x)$. 

In all three cases above, we have $q^{(R')}(x) - q^{(R)}(x) \ge q^{(T')}(x) - q^{(T)}(x)$. We also have $q^{(R)}(x) \le q^{(T)}(x)$ (shown earlier). Thus by concavity we conclude that $G_{D_f}(p, q^{(R')}) - G_{D_f}(p, q^{(R)}) \ge G_{D_f}(p, q^{(T')}) - G_{D_f}(p, q^{(T)})  \implies F_G(R \cup \{(a,b)\}) - F_G(R) \ge F_G(T \cup \{(a,b)\}) - F_G(T)$.
\end{proof}

\hoverlap*
\begin{proof}[Construction]
We begin by considering overlap measures of the form $G(p,q) = \sum_{x\in \Omega} g_1(p(x)) \cdot g_2(q(x))$ for nonnegative functions $g_1$ and $g_2$. Given a non-decreasing concave $g_2$, we fully specify the overlap measure by choosing $g_1$ such that $G$ is uniquely maximized when $q=p$.

That is, we consider the constrained maximization of $\sum_{i=1}^g g_1(p_i) \cdot g_2 (q_i)$, subject to $\sum_{i=1}^g q_i \le 1$. By placing a Lagrange multiplier of $\lambda$ on the constraint, we see that the maximum occurs when $g_1 (p_i) \cdot g_2' (q_i) = \lambda$ for all $i$. Since we would like this to be satisfied when $q_i = p_i$ for all $i$, and we can scale the overlap measure by a multiplicative constant without loss, it suffices to set $g_1$ identically to $\frac{1}{g_2'}$. Rewriting using $g_2 = h$ gives the overlap measure $G(p,q) = \sum_{x\in \Omega} \frac{h(q(x))}{h'(p(x))}$. Since $p$ is given as a fixed distribution and $h$ is non-decreasing, it is clear that $G$ is non-decreasing in all $q(x)$. Now, we will show that $G$ is MDR by analyzing $F_G$.
\end{proof}

\begin{proof}[Proof of monotonicity] For any set $R$ and each item $j$, we define $\ell_R (i)\coloneqq \min \left\{ j\in [k] |  (i,j) \in R \right\}$ (or $\ell_R(i) = k+1$ if no such $j$ exists) as the earliest position in which $R$ places item $i$. We also define $w_{k+1} = 0$. 

 Observe that $R^{\le j}  \setminus R^{\le j-1}$ is exactly the set of items which appear for the \emph{first} time in position $j$; thus $\ell_R(i) = j$ for all $i \in R^{\le j} \setminus R^{\le j-1}$. Additionally, $R^{\le 1} \subseteq R^{\le 2} \subseteq \dots \subseteq R^{\le k}$. Then, we may write

\begin{align*}
    F_G(R) &= \sum_{x \in \Omega} \frac{1}{h'(p(x))} \cdot h\left(\sum_{j=1}^{k} w_j \left(\sum_{i\in R^{\le j} \setminus R^{\le j-1}} q_i(x) \right) \right) \\
    &= \sum_{x \in \Omega} \frac{1}{h'(p(x))} \cdot h\left(\sum_{j=1}^k \sum_{i \in R^{\le j} \setminus R^{\le j-1}} w_{\ell_R(i)} q_{i}(x) \right) \\
    &= \sum_{x \in \Omega} \frac{1}{h'(p(x))} \cdot h\left(\sum_{i \in R^{\le k}} w_{\ell_R(i)} q_{i}(x)\right).
\end{align*}

Now, consider a superset $T \supseteq R$. It is clear from the definition that $\ell_T(i) \le \ell_R(i)$ for all $i$, so $w_{\ell_T(i)} \ge w_{\ell_R(i)}$. Further, $R^{\le k} \subseteq T^{\le k}$. Then, for any $x$,
$$\sum_{i \in R^{\le k}} w_{\ell_R(i)} q_i(x) \le \sum_{i \in T^{\le k}} w_{\ell_R(i)} q_i(x) \le \sum_{i \in T^{\le k}} w_{\ell_T(i)} q_i(x).$$ 
Since $h$ is nonnegative and non-decreasing, the inequality is preserved by applying $h$ and then multiplying by $\frac{1}{h'(p(x))}$. Summing over all $x\in\Omega$ gives $F_G(R) \le F_G(T)$.
\end{proof}

\begin{proof}[Proof of submodularity]
We now show that for any sets $R\subseteq T$ and element $(a,b)$, $F_G(R \cup \{(a,b)\}) - F_G(R) \ge F_G(T\cup \{(a,b)\}) - F_G(T)$. If $(a,b) \in T$, the inequality clearly holds since $F_G$ is monotone, so suppose $(a,b) \notin T$. To simplify notation, we also write $R' = R\cup \{(a,b)\}$, $T' = T\cup \{(a,b)\}$. Observe that for all $i\ne a$, $\ell_{R'}(i) = \ell_R(i) \ge \ell_T(i) = \ell_{T'}(i)$. We also have $\ell_{R'}(a) \le \ell_R(a)$ and $\ell_{T'}(a) \le \ell_T(a)$. We now take cases:

\paragraph*{Case 1: $a \in R^{\le k} (\subseteq T^{\le k})$.}

Then $(R')^{\le k} = R^{\le k}$ and $(T')^{\le k} = T^{\le k}$. We take subcases based on $\ell_{T}(a)$, and show that in both subcases, $w_{\ell_{R'}(a)} - w_{\ell_R(a)} \ge w_{\ell_{T'}(a)} - w_{\ell_T(a)}$.
\begin{itemize}
    \item $\ell_T(a) \le b$: Then $\ell_{T'}(a) = \ell_T(a)$, while $\ell_{R'}(a) \le \ell_R(a)$. So $w_{\ell_{R'}(a)} - w_{\ell_{R}(a)} \ge 0 = w_{\ell_{T'}(a)} - w_{\ell_{T}(a)}$.
    
    \item $b< \ell_T(a)$ ($\le \ell_R(a)$): Then $\ell_{R'}(a) = \ell_{T'}(a) = b$, so $w_{\ell_{R'}(a)} - w_{\ell_R(a)} = w_b -  w_{\ell_R(a)} \ge w_b -  w_{\ell_T(a)} \ge w_{\ell_{T'}(a)} - w_{\ell_T(a)}$.
\end{itemize}
So, we have
\begin{align*}
&h\left(\sum_{i \in (R')^{\le k}} w_{\ell_{R'} (i)} q_i(x)\right) - h\left(\sum_{i \in R^{\le k}} w_{\ell_R (i)} q_i(x)\right) \\
&= h\left(\sum_{i \in R^{\le k}} w_{\ell_{R'} (i)} q_i(x)\right) - h\left(\sum_{i \in R^{\le k}} w_{\ell_R (i)} q_i(x)\right) \\
&= h\left((w_{\ell_{R'} (a)} - w_{\ell_{R} (a)}) q_a(x) + \sum_{i \in R^{\le k}} w_{\ell_{R} (i)} q_i(x)\right) \\
    &\quad- h\left(\sum_{i \in R^{\le k}} w_{\ell_R (i)} q_i(x)\right) \\
&\ge h\left((w_{\ell_{T'} (a)} - w_{\ell_{T} (a)}) q_a(x) + \sum_{i \in R^{\le k}} w_{\ell_{R} (i)} q_i(x)\right) \\
    &\quad- h\left(\sum_{i \in R^{\le k}} w_{\ell_R (i)} q_i(x)\right) \\
&\ge h\left((w_{\ell_{T'} (a)} - w_{\ell_{T} (a)}) q_a(x) + \sum_{i \in T^{\le k}} w_{\ell_{T} (i)} q_i(g)\right) \\
    &\quad- h\left(\sum_{i \in T^{\le k}} w_{\ell_T (i)} q_i(x)\right) \\
&=h\left(\sum_{i \in (T')^{\le k}} w_{\ell_{T'} (i)} q_i(x)\right) - h\left(\sum_{i \in T^{\le k}} w_{\ell_T (i)} q_i(x)\right),
\end{align*}
where the first inequality is from the subcase analysis (and monotonicity of $h$) and the second inequality holds by concavity of $h$.

\paragraph*{Case 2: $a\notin R^{\le k}$, $a\in T^{\le k}$.}

In this case, $(R')^{\le k} = R^{\le k} \cup \{a\}$ and $(T')^{\le k} = T^{\le k}$. We have $\ell_{R'}(a) = b$. If $\ell_T(a) < b$, then $\ell_{T'}(a) = \ell_T(a)$, so $w_{\ell_{T'}(a)} - w_{\ell_T(a)} = 0$. If $\ell_T(a) \ge b$, then $\ell_{T'}(a) = b$, so $w_{\ell_{T'}(a)} - w_{\ell_T(a)} = w_b - w_{\ell_T(a)}$. In either case, we have $w_{\ell_{T'}(a)} - w_{\ell_T(a)} \le w_b$, so we have 
\begin{align*}
&h\left(\sum_{i \in (R')^{\le k}} w_{\ell_{R'} (i)} q_i(x)\right) - h\left(\sum_{i \in R^{\le k}} w_{\ell_R (i)} q_i(x)\right) \\
&= h\left(w_b q_a(x) + \sum_{i \in R^{\le k}} w_{\ell_{R} (i)} q_i(x)\right) - h\left(\sum_{i \in R^{\le k}} w_{\ell_R (i)} q_i(x)\right) \\
&\ge h\left((w_{\ell_{T'}(a)} - w_{\ell_T(a)}) q_a(x) + \sum_{i \in R^{\le k}} w_{\ell_{R} (i)} q_i(x)\right) \\
    &\quad- h\left(\sum_{i \in R^{\le k}} w_{\ell_R (i)} q_i(x)\right) \\
&\ge h\left((w_{\ell_{T'}(a)} - w_{\ell_T(a)}) q_a(x) + \sum_{i \in T^{\le k}} w_{\ell_{T} (i)} q_i(x)\right) \\
    &\quad- h\left(\sum_{i \in T^{\le k}} w_{\ell_T (i)} q_i(x)\right) \\
&=h\left(\sum_{i \in (T')^{\le k}} w_{\ell_{T'} (i)} q_i(x)\right) - h\left(\sum_{i \in T^{\le k}} w_{\ell_T (i)} q_i(x)\right),
\end{align*}
where the first inequality holds due to the subcase analysis, and the second inequality is due to concavity of $h$. 

\paragraph*{Case 3: $a \notin T^{\le k} (\supseteq R^{\le k})$.}

In this case, $(R')^{\le k} = R^{\le k} \cup \{a\}$ and $(T')^{\le k} = T^{\le k} \cup\{a\}$, and $\ell_{R'}(a) = \ell_{T'}(a) = b$. Then,
\begin{align*}
&h\left(\sum_{i \in (R')^{\le k}} w_{\ell_{R'} (i)} q_i(x)\right) - h\left(\sum_{i \in R^{\le k}} w_{\ell_R (i)} q_i(x)\right) \\
&= h\left(w_b q_a(x) + \sum_{i \in R^{\le k}} w_{\ell_{R} (i)} q_i(x)\right) - h\left(\sum_{i \in R^{\le k}} w_{\ell_R (i)} q_i(x)\right) \\
&\ge h\left(w_b q_a(x) + \sum_{i \in T^{\le k}} w_{\ell_{T} (i)} q_i(x)\right) - h\left(\sum_{i \in T^{\le k}} w_{\ell_T (i)} q_i(x)\right) \\
&=h\left(\sum_{i \in (T')^{\le k}} w_{\ell_{T'} (i)} q_i(x)\right) - h\left(\sum_{i \in T^{\le k}} w_{\ell_T (i)} q_i(x)\right),
\end{align*}
where the inequality is due to concavity of $h$.

In all cases, we finish by multiplying by $\frac{1}{h'(p(x))}$ and summing over all $x\in\Omega$ to give $$F_G(R\cup \{(a,b)\}) - F_G(R) \ge F_G(T\cup \{(a,b)\}) - F_G(T).$$
\end{proof}

\subsection{Proofs and discussion from Section~\ref{sec:distributional}} \label{sec:proofsdistributional}

\paragraph*{A $(1-1/e)$-approximation with repeated elements.} If we permit our recommendation list to include repeated elements, a sequence is simply an assignment of (at most) one item to each position. Then, a $(1-1/e)$-approximation algorithm is possible using a subroutine of maximizing submodular functions over a matroid due to \cite{calinescu2011}.

More formally, consider the partition of the ground set into the sets $V = \bigcup_{\ell \in [k]} E_\ell$, where $E_\ell \coloneqq \{(i,j) \mid  i \in [K], j = \ell\}$, and the independent sets $\mathcal{I} \coloneqq \left\{R\subseteq V \mid |R\cap E_\ell| \le 1, \; \forall \ell \in [k] \right\}$. 

The set function $F_G(R)$ as defined in Section \ref{sec:mdr_overlap} only considers each item to contribute to the distribution at the first position that it occurs. If we do want to consider sequences with repeated elements, an alternate definition is more useful:
\[ \hat F_G(R) \coloneqq G\left(p, \sum_{j=1}^{k} w_j \left(\sum_{i\in R^j } q_i\right)\right), \]
where $R^j=\{i \in [K] \mid (i,j)\in R\}$. For this subsection, we assume that $\hat{F}_G$ is monotone and submodular, which is indeed the case for each of the applications discussed in Section \ref{sec:overlap_families}. Since $\hat F_G$ is a monotone submodular set function, the continuous greedy algorithm and pipage rounding technique of \cite{calinescu2011} for submodular function maximization subject to a single matroid constraint yields a $(1-1/e)$-approximate independent set, which will be a basis of the matroid due to monotonicity. There is a straightforward bijection between sequences and bases of the matroid $(V, \mathcal{I})$, and maximizing $G$ over sequences is equivalent to maximizing $\hat F_G$ subject to a partition matroid constraint.

However, in a more realistic application, we would prefer our recommendation list not to repeat items (repeating a single movie many times would hardly constitute a ``diverse'' recommendation list, no matter how heterogeneous the distribution of that particular movie). At the same time, each item's precise genre breakdown is likely to be unique, so merely replacing a repeated item with a close substitute may be impossible. Instead, we explicitly forbid repeats, meaning that independent sets of the partition matroid described above no longer all correspond to legal sequences. Although the partition matroid reduction no longer works, the set of SMDR overlap measures enables an alternate laminar matroid reduction which recovers the $(1-1/e)$ guarantee.

\settoseq*
\begin{proof}
For every item $i$, we define $\ell_R (i)$ to be the first position that $i$ occurs in $R$; that is, $\ell_R (i)\coloneqq \min \left\{ j\in [k] |  (i,j) \in R \right\}$, or $\ell_R(i) = k+1$ if no such $j$ exists. We also introduce the notation of $w_{k+1}=0$. Sort the items in increasing order of $\ell_R(\cdot)$ (breaking ties arbitrarily), and call this sequence $\pi$. We claim that $G(\pi) \ge F_G(R)$.

Consider an arbitrary item $j$. By definition of the laminar matroid, 
$$|R \cap D_{\ell_R(i)}| = \sum_{y = 1}^k \sum_{x=1}^{\ell_R (i)} \1{(x,y) \in R} \le \ell_R(i). $$ 
The summation is an upper bound on the number of items $x$ with $(x,y) \in R$ for some $y \le \ell_R (i)$. But these are exactly the items with $\ell_R(x) \le \ell_R(i)$ (including $i$ itself), and therefore the items that can appear before $i$ in $\pi$. So the position at which $i$ appears in $\pi$, denoted $\pi^{-1}(i)$, is less than or equal to $\ell_R(i)$. This implies $w_{\pi^{-1}(i)} \ge w_{\ell_R (i)}$ for all $i$. 

Now, observe that $R^{\le j} \setminus R^{\le j-1}$ is exactly the set of items which appear for the \emph{first} time in position $j$; thus $\ell_R(i) = j$ for all $i \in R^{\le j} \setminus R^{\le j-1}$. Additionally, $R^{\le 1} \subseteq R^{\le 2} \subseteq \dots \subseteq R^{\le k} \subseteq [K]$. Then, for any genre $g$, we have
\begin{align*}
\sum_{j=1}^{k} w_j \left(\sum_{i\in R^{\le j} \setminus R^{\le j-1}} q_i(g) \right) &= \sum_{j=1}^k \sum_{i \in R^{\le j} \setminus R^{\le j-1}} w_{\ell_R(i)} q_{i}(g) \\
&= \sum_{i \in R^{\le k}} w_{\ell_R(i)} q_{i}(g) \\ 
&\le \sum_{i \in [K]} w_{\pi^{-1}(i)} q_{i}(g) \\
&= \sum_{j=1}^k  w_j q_{\pi(j)} (g).
\end{align*}

Then, since $G$ is non-decreasing with respect to all $q(g)$, we conclude that $G(R) \le F_G(\pi)$. 
\end{proof}

\maxGf*
\begin{proof}
Let $\pi^* \coloneqq \argmax_{\pi} G(\pi)$, and define $$R^* \coloneqq \left\{ {((\pi^*)^{-1}(j),j)} \mid j \in [k] \right\}$$ as the set corresponding to the item-position pairs in $\pi^*$.

By construction, $F_G(R^*) = G(\pi^*)$, and $R^* \in \mathcal{I}$. Then, by monotonicity the the maximum value over independent sets is attained by a basis, and we get
\[\max_{R \in \mathcal{I}} F_G(R)  \ge F_G(R^*) = G(\pi^*) = \max G(\pi). \]
\end{proof}

Finally, one might wonder if it would be possible to take a solution obtained from the problem with allowed repeats via a partition matroid and convert it to a solution without repeats, simply by showing items in the order of the first time they appear. The barrier to simply taking the first occurrence of each repeated item is that this operation may destroy the submodular structure of the original function, so that the continuous greedy algorithm no longer provides a near-optimal approximation guarantee. That is, consider a submodular set function $f(S)$ on the set of item-position pairs, and the function $\bar f(S)$ that evaluates $f$ on the subset of $S$ corresponding to the first occurrence of each item. The following simple example illustrates that $\bar f(S)$ need not be submodular:

Denote the items as $a$ and $b$ and the positions as $1,2,3$, and explicitly define the value of $f$ on the following subsets:
    $$f(a,1)=4, f(a,2)=2, f(b,3)=2,$$
    $$f(a,1; b,3)=6, f(a,1; a,2) = 6, f(a,2;b,3)=3,$$
    $$f(a,1; a,2,b,3)=6.$$
It is straightforward to check that for these values, $f$ is submodular. Now observe that $\bar f$ takes values  
$$ \bar f(a,1;a,2;b,3) = f(a,1; b,3) = 6, $$
$$\bar f(a,1;a,2) = f(a,1) = 4 .$$

But then 
$$ \bar f(a,2;b,3) - \bar f(a,2) = 1 < 2 = \bar f(a,1;a,2;b,3) - \bar f(a,1;a,2),$$

So $\bar f$ is not submodular. Therefore, the new laminar matroid construction is indeed necessary to avoid this issue.

\subsection{Proofs from Section~\ref{sec:ordsub}} \label{sec:proofsordsub}

\proptwoapprox*
\begin{proof}
    Denote the sequence of length $k$ maximizing $f$ as $S = s_1 s_2 \dots s_k$ and the sequence of length $k$ maximizing the marginal increase at each step as $A = a_1 a_2 \dots a_k$. We write $S^j = s_j s_{j+1} \dots s_k$ to denote the suffix of $S$ starting at element $s_j$. 
    
    Let $OPT(k) = f(S)$, $ALG(k) = f(A)$, so that we seek to show that $ALG(k) \ge \frac{1}{2} OPT(k)$ for all $k$. We bound the performance of the greedy algorithm by comparing it to the optimal solution. The key insight is to ask the following question at each step: if we must remain committed to all the greedily chosen elements so far, but make the same choices as the optimum for the rest of the elements, how much have we lost? 
    
    To answer this question, we show via induction that for all $i$, $f(A_i || S^{i+1}) \ge OPT(k) - f(A_i).$
    
    The base case of $i=0$ is trivial, as $f(A_0 || S^1) = f(S) = OPT(k) \ge OPT(k) - f(A_0)$. So suppose the claim is true for some $i-1$, and observe that by ordered submodularity we have 
    \begin{align*}
        f(A_{i-1} || s_i) - f(A_{i-1}) &\ge f(A_{i-1} || s_i || S^{i+1}) - f(A_{i-1} || a_i || S^{i+1}) \\
        &= f(A_{i-1} || S^i) - f(A_i || S^{i+1}).
    \end{align*}
    Rearranging and applying the choice of the greedy algorithm, by which $f(A_i) \ge f(A_{i-1} || s_{i})$, gives 
    \begin{align}\label{eq:approx-ineq}
        f(A_{i} || S^{i+1}) \ge f(A_{i-1} || S^{i}) + f(A_{i-1}) - f(A_{i}).
    \end{align}
    Now applying the induction hypothesis yields
    \begin{align*}
        f(A_{i} || S^{i+1}) &\ge (OPT(k) - f(A_{i-1})) + f(A_{i-1}) - f(A_{i}) \\
        &= OPT(k) - f(A_{i}),
    \end{align*}
    completing the induction.
    
    Finally, taking $i=k$ in the claim gives $f(A) \ge OPT(k) - f(A) \implies ALG(k) \ge \frac{1}{2} OPT(k).$
\end{proof}

\ordsuboverlap*
\begin{proof}
    Consider the $\hat F_G$ function defined in Section \ref{sec:proofsdistributional}, and define the following two sets of item-position pairs:
    \begin{align*}
        R &= \{ (s_j, j) \mid j \in [i-1] \}, \\
        T &= \{ (s_j, j) \mid j \in [i-1] \cup [i+1,k] \}.
    \end{align*}
    Now observe that by construction,
    \begin{align*}
    G(s_1 \dots s_i) - G(s_1 \dots s_{i-1}) &= \hat F_G(R \cup \{(s_i, i)\}) - \hat F_G(R) \\
    &\ge \hat F_G(T \cup \{(s_i, i)\}) - \hat F_G(T) \\
    &\ge \hat F_G(T \cup \{(s_i, i)\}) - \hat F_G(T \cup \{(\bar{s}_i, i)\}) \\
    &= G(s_1 \dots s_i \dots s_k) \\
        &\quad- G(s_1 \dots
s_{i-1} \bar{s}_i s_{i+1} \dots s_k),
    \end{align*}
    where the first inequality is due to submodularity of $\hat F_G$ and the second inequality is due to monotonicity of $\hat F_G$ (since $G$ is MDR). Thus $G$ is ordered-submodular.
\end{proof}

\subsection{Proofs from Section~\ref{sec:cal2/3}} \label{sec:proofscal2/3}

\strongerineq*
\begin{proof}
Recall that we assume that $g' \ne g^*$ and $f(A_i || T^{i+1}) \le f(A_{i-1} || T^i)$. We seek to construct an assignment $\bar T^{i+1}$ of the remaining slots $i+1$ from $k$ such that $f(A_i || \bar T^{i+1}) \ge f(A_{i-1} || T^{i}) - \frac{1}{2} (f(A_i) - f(A_{i-1})$, starting from $T^{i+1}$ (the assignment formed by simply dropping $t_i$, the first item of $T^i$). Depending on the size of $\tau(g')$ relative to $w_i$, we modify $T^i$ with a different approach. 

\textbf{\emph{Case 1:} $\tau(g') \ge \frac{1}{2} w_i$.}

We construct the desired $\bar T^{i+1}$ by starting from $T^{i+1}$ and ``correcting'' by moving subsequent weights from $g'$ to $g^*$. We claim that we can always correct by an amount that is at least $\frac{w_i}{2}$ and at most $w_i$. 

\emph{Case 1a:} $\tau(g') \ge w_i$. Either the next weight in $g'$ is at least $\frac{w_i}{2}$ (but by definition at most $w_i$), so we can just move this weight and be done, or otherwise $\tau(g')$ is composed of many weights less than $\frac{w_i}{2}$. Adding these weights one at a time in descending order, we must be able to stop somewhere between $\frac{w_i}{2}$ and $w_i$ (if the total before adding a weight is less than $\frac{w_i}{2}$, and the total after adding that weight is greater than $w_i$, then that weight must have been greater than $\frac{w_i}{2}$, which is a contradiction). So, moving the weights up to this stopping point creates a total correction between $\frac{w_i}{2}$ and $w_i$.

\emph{Case 1b:} $\tau(g') \in \left[\frac{w_i}{2}, w_i\right)$. In this case, we simply correct by the entirety of $\tau(g')$ (that is, move \emph{all} subsequent weights from $g'$ to $g^*$). 

So in either subcase, we can correct by some amount $z \in \left[\frac{w_i}{2}, w_i\right]$.  Now, consider the function 
\begin{multline*}
c(x) = \sqrt{p(g')} \sqrt{\alpha(g') + \tau(g') + w_i - x} \\ + \sqrt{p(g^*)} \sqrt{\alpha(g^*) + \tau(g^*) + x},
\end{multline*}
representing the contribution from genres $g'$ and $g^*$ towards $f$, after correcting by $x$ (it will suffice to consider only these two genres, since all other genres are assigned the same slots in both $T^{i+1}$ and $\bar T^{i+1}$). Observe that $x=0$ corresponds to $f(A_i || T^{i+1})$ and $x=w_i$ corresponds to $f(A_{i-1} || T^i)$, so $c(0) \le c(w_i)$. Taking second derivatives, we get 
\begin{align*}
c''(x) &= -\frac{\sqrt{p(g')}}{4\left(\alpha(g') + \tau(g') + w_i - x\right)^{3/2}}  \\
    &\quad - \frac{\sqrt{p(g^*)}}{4\left(\alpha(g^*) + \tau(g^*) + x\right)^{3/2}} \\
&< 0,
\end{align*}
so $c$ is a positive, concave function of $x$. Suppose the (continuous) maximum of $c$ occurs at $x^*$. 

If $w_i \le x^*$, then first by monotonicity and then by concavity we have
\begin{align*}
    c(w_i) - c(z) &\le c(w_i) - c(w_i / 2) \\
    &\le \frac{1}{2} \left(c(w_i) - c(0)\right), \\
    \implies f(A_{i-1} || T^i) - f(A_{i} || \bar T^{i+1}) &\le \frac{1}{2} \left(f(A_{i-1} || T^i) - f(A_{i} || T^{i+1}) \right) \\
    &\le \frac{1}{2} \left(f(A_{i}) - f(A_{i-1}) \right),
\end{align*}
where the first line is because a correction that is at least $\frac{1}{2} w_i$ increases $f$ by at least half the amount that a full correction of $w_i$ would have achieved (as depicted in the figure below), and the final inequality is an application of inequality (\ref{eq:approx-ineq}) as re-proved earlier. 

\begin{figure}[ht] \label{fig:concaveAPP}
    \centering
    \begin{tikzpicture}[scale=0.8]
    \begin{axis}[
            ticks=none,
            axis lines=middle,
            xlabel={$x$},
            xmin=0, xmax=9,
            ymin=0, ymax=4.2,
            clip=false,
            xticklabels=\empty,
            yticklabels=\empty
        ]
        \addplot+[mark=none,samples=200,unbounded coords=jump, domain=0.5:8.5] {ln(x)+1.5};

        \draw[fill] (axis cs:1,1.5) circle [radius=1.5pt] node[below right] {};
        \draw[fill] (axis cs:4,2.8863) circle [radius=1.5pt] node[below right] {};
        \draw[fill] (axis cs:7,3.4459) circle [radius=1.5pt] node[below right] {};

        \draw (axis cs:1,0) node[below] {$0$};
        \draw (axis cs:4,0) node[below] {$z \ge \nicefrac{w_i}{2}$};
        \draw (axis cs:7,0) node[below] {$w_i$};
        \draw (axis cs:0,1.5) node[left] {$f(A_{i} || T^{i+1})$};
        \draw (axis cs:0,2.8863) node[left] {$f(A_{i} || \bar T^{i+1})$};
        \draw (axis cs:0,3.4459) node[left] {$f(A_{i-1} || T^i)$};
        
        \draw[dashed] (axis cs:1,0) -- (axis cs:1,1.5) -- (axis cs:0,1.5);
        \draw[dashed] (axis cs:4,0) -- (axis cs:4,2.8863) -- (axis cs:0,2.8863);
        \draw[dashed] (axis cs:7,0) -- (axis cs:7,3.4459) -- (axis cs:0,3.4459);
    \end{axis}
    \end{tikzpicture}
    \caption{Change in $f(A_i||\bar T^{i+1})$ as we move $x$ weight from $g'$ to $g^*$.}
\end{figure}
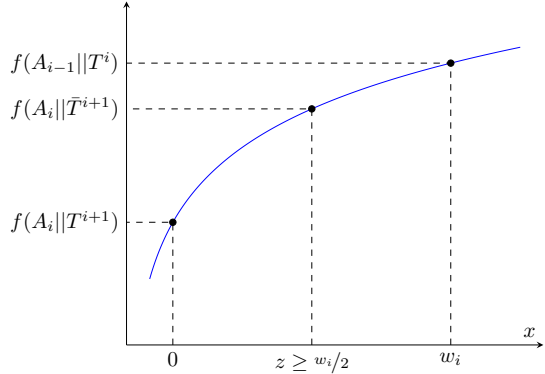

If $w_i > x^*$, then by continuity there must exist some $w' \in [0, x^*]$ such that $g(w') = g(w_i)$. Observe that $z \ge \frac{w_i}{2} > \frac{x^*}{2} \ge \frac{w'}{2}$, and then the same argument holds:
\begin{align*}
    c(w') - c(z) \le c(w') - c(w' / 2) &\le \frac{1}{2} \left(c(w') - c(0)\right) \\
    &= \frac{1}{2} \left(c(w_i) - g(0)\right), \\
    \implies f(A_{i-1} | T^i) - f(A_{i} | \bar T^{i+1} ) &\le \frac{1}{2} \left(f(A_{i}) - f(A_{i-1}) \right).
\end{align*}

Either way, rearranging establishes the desired inequality: 
\[f(A_i | \bar T^{i+1}) \ge f(A_{i-1} | T^i) - \frac{1}{2} \left(f(A_i) - f(A_{i-1})\right). \]

\textbf{\emph{Case 2:} $\tau(g') < \frac{1}{2} w_i$.}

Since $\tau(g')$ is small, we do not have much that we can move to compensate. However, the fact that $\tau(g')$ is so small means that even without correcting, $f(A_i || T^{i+1})$ is not as bad as inequality (\ref{eq:approx-ineq}) suggests it could be.

Again noticing that we only need to focus on the difference in contributions to $f$ from $g'$ and $g^*$ (since all other genre assignments are unchanged), we have
\begin{align}
    \begin{split}
     &\frac{f(A_{i-1} || T^i) - f(A_i || T^{i+1})}{f(A_i) - f(A_{i-1})} \\
     &= \frac{\sqrt{p(g')} \left(\sqrt{\alpha(g') + \tau(g')} - \sqrt{\alpha(g') + w_i + \tau(g')}\right) }{\sqrt{p(g')} \left(\sqrt{\alpha(g') + w_i} - \sqrt{\alpha(g')}\right)} \\
    & \qquad + \frac{\sqrt{p(g^*)} \left(\sqrt{\alpha(g^*) + w_i + \tau(g^*)} - \sqrt{\alpha(g^*) + \tau(g^*)}\right)}{\sqrt{p(g')} \left(\sqrt{\alpha(g') + w_i} - \sqrt{\alpha(g')}\right)} \\
    \end{split} \\
    \begin{split}
    & \le \frac{\sqrt{p(g')} \left(\sqrt{\alpha(g') + \tau(g')} - \sqrt{\alpha(g') + w_i + \tau(g')}\right) }{\sqrt{p(g')} \left(\sqrt{\alpha(g') + w_i} - \sqrt{\alpha(g')}\right)} \\
    & \qquad + \frac{\sqrt{p(g^*)} \left(\sqrt{\alpha(g^*) + w_i} - \sqrt{\alpha(g^*)}\right)}{\sqrt{p(g')} \left(\sqrt{\alpha(g') + w_i} - \sqrt{\alpha(g')}\right)} \\
    \end{split}\\
    \begin{split}
    & \le \frac{\sqrt{p(g')} \left(\sqrt{\alpha(g') + \tau(g')} - \sqrt{\alpha(g') + w_i + \tau(g')}\right) }{\sqrt{p(g')} \left(\sqrt{\alpha(g') + w_i} - \sqrt{\alpha(g')}\right)} \\
    & \qquad + \frac{\sqrt{p(g')} \left(\sqrt{\alpha(g') + w_i} - \sqrt{\alpha(g')}\right)}{\sqrt{p(g')} \left(\sqrt{\alpha(g') + w_i} - \sqrt{\alpha(g')}\right)}, \\
    \end{split}
\end{align}

where (4) is due to concavity of the square root (so we take $\tau(g^*) = 0$) and (5) is due to the definition of greedy choosing $g'$ over $g^*$. Continuing to simplify, we have
\begin{align}
    \begin{split}
    &\frac{f(A_{i-1} || T^i) - f(A_i || T^{i+1})}{f(A_i) - f(A_{i-1})} \\
    &\le 1 - \frac{\sqrt{\alpha(g') + w_i + \tau(g')} - \sqrt{\alpha(g') + \tau(g')}}{\sqrt{\alpha(g') + w_i} - \sqrt{\alpha(g')}} 
    \end{split} \\
    &\le 1 - \frac{\sqrt{\alpha(g') + 3w_i/2} - \sqrt{\alpha(g') + w_i/2}}{\sqrt{\alpha(g') + w_i} - \sqrt{\alpha(g')}} \\
    &\le 1 - \frac{\sqrt{3w_i/2} - \sqrt{w_i/2}}{\sqrt{w_i} - \sqrt{0}} \\
    &= 1 - \sqrt{2 - \sqrt{3}} \approx 0.48 < \frac{1}{2}.
\end{align}
Here, (6) is a simplified form of (4), and (7) is again due to concavity (so we take $\tau(g') = \frac{w_i}{2})$. Then, $\frac{\sqrt{x+3w/2} - \sqrt{x+w/2}}{\sqrt{x+w} - \sqrt{x}}$ is an increasing function of $x$, since its first derivative with respect to $x$ is
\begin{align*}
    &\frac{\frac{\sqrt{x+3w/2} - \sqrt{x + w/2}}{\sqrt{x}{\sqrt{x+w}}} + \frac{1}{\sqrt{x + 3w/2}} - \frac{1}{\sqrt{x+w/2}}}{2(\sqrt{x + w} - \sqrt{x})} \\
    &= \frac{\frac{\sqrt{x+3w/2} - \sqrt{x + w/2}}{\sqrt{x}{\sqrt{x+w}}} + \frac{\sqrt{x+w/2} - \sqrt{x+3w/2}}{\sqrt{x+w/2} \sqrt{x + 3w/2}}}{2(\sqrt{x + w} - \sqrt{x})} \\
    &= \frac{\sqrt{x+3w/2} - \sqrt{x+w/2}}{2(\sqrt{x+w} - \sqrt{x})} \left(\frac{1}{\sqrt{x}{\sqrt{x+w}}} - \frac{1}{\sqrt{x+w/2}\sqrt{x+3w/2}} \right) \\
    &> 0,
\end{align*}
so we may take $\alpha(g') = 0$ to get (8). Finally, rearranging (9) results in 
\begin{align*}
    f(A_{i-1} || T^i) - f(A_{i}  || T^{i+1}) &\le \frac{1}{2} \left(f(A_i) - f(A_{i-1})\right) \\
    \implies f(A_{i}  || T^{i+1})&\ge f(A_{i-1} || T^i) - \frac{1}{2} \left(f(A_i) - f(A_{i-1})\right),
\end{align*}
so we may just take $T^{i+1}$ to be our $\bar T^{i+1}$ without any correction at all, which establishes the desired \[f(A_i | \bar T^{i+1}) \ge f(A_{i-1} | T^i) - \frac{1}{2} \left(f(A_i) - f(A_{i-1})\right). \]
\end{proof}

\twothirds*
\begin{proof}
We define $S^{(1)}$ to be the optimal sequence $S$, and using Lemma~\ref{lem:ineq1/2} with $T^i = S^{(i)}$, we define inductively $S^{(i+1)} = \bar T^{i+1}$ (the existence of which is stated by the lemma).

We show via induction that for all $i$, $$f(A_i || S^{(i+1)}) \ge OPT(k) - \frac{1}{2}f(A_i).$$ 

    For the base case of $i=0$, we have $f(A_0 || S^{(1)}) = f(S) = OPT(k) \ge OPT(k) - \frac{1}{2}f(A_0)$. So suppose the claim holds for some $i$, and observe that by Lemma~\ref{lem:ineq1/2} and the fact that $f(A_{i+1}) \ge f(A_i || s_{i+1})$ by definition of the greedy algorithm, we have
    \begin{align*}
        f(A_{i+1} || S^{(i+2)}) &\ge f(A_i || S^{(i+1)}) - \frac{1}{2} (f(A_{i+1}) - f(A_{i})) \\
        &\ge OPT(k) - \frac{1}{2} f(A_i) - \frac{1}{2} (f(A_{i+1}) - f(A_{i})) \\
        &= OPT(k) - \frac{1}{2} f(A_{i+1}),
    \end{align*}
    completing the induction. Finally, setting $i=k$ establishes $ALG(k) \ge \frac{2}{3} OPT(k)$.
\end{proof}

\section{Mixed sign issues with KL divergence-based overlap measures} \label{sec:bad_netflix}
A natural hope might be to use the KL divergence as a calibration heuristic, as it is perhaps the most commonly used statistical divergence. Unfortunately, the KL divergence cannot be used directly because it is unbounded; our translation to the distance-based overlap measure is also not well-defined on the KL divergence for the same reason. In~\cite{Steck18} an alternative transformation is proposed, yielding the following calibration heuristic: 
$$f(\mathcal{I}) = \sum_g p(g|u) \log \sum_{i \in \mathcal{I}} w_{r(i)} \tilde{q} (g|i).$$
However, this objective function has inconsistent sign, depending on how the recommendation weights are chosen (and we note that Steck does not set any constraints on the weights), and consequently the greedy choice can be far from optimal. In fact, we show that the greedy solution can be negative, while the optimum is positive. So the KL divergence (and variants of it) are not conducive to multiplicative approximation guarantees for the calibration problem.

Suppose there are $4$ genres ($g_k$ for $k=1,2,3,4$), $2$ movies ($i_\ell$ for $\ell=1,2$), and $1$ user ($u$), and that we seek a recommendation list of length $2$ with weights $w_1 > w_2 = 1$. For simplicity of notation, we denote $p(g_k|u)$ as $p_k$. Suppose further that the movies have the following distributions over genres for some $\varepsilon \in (0,\frac{1}{3})$:

\begin{center}
\begin{tabular}{ p{4cm} p{4cm} }
$\tilde{q}(g_1 | i_1) = \frac{1}{2}(1-\varepsilon)$ & $\tilde{q}(g_1 | i_2) = \frac{1}{2}(1-\varepsilon)$ \\[6pt]
$\tilde{q}(g_2 | i_1) = \frac{1}{4}(1-\varepsilon)$ & $\tilde{q}(g_2 | i_2) = \frac{1}{2}(1-\varepsilon)$ \\[6pt]
$\tilde{q}(g_3 | i_1) = \frac{1}{4}(1-\varepsilon)$ & $\tilde{q}(g_3 | i_2) = \frac{\varepsilon}{2}$ \\[6pt]
$\tilde{q}(g_4 | i_1) = \varepsilon$ & $\tilde{q}(g_4 | i_2) = \frac{\varepsilon}{2}$ 
\end{tabular}
\end{center}

Finally, suppose the parameters are such that $$p_3 \log\left(\frac{1-\varepsilon}{2\varepsilon} \right) = (p_2-p_4) \log \left(\frac{2w_1 + 1}{w_1+2} \right).$$

Then, observe that 
\begin{align*}
    f(i_1 i_2) - f(i_2 i_1) &= p_2 \log \left( \frac{\frac{w_1}{4}(1-\varepsilon) + \frac{1}{2}(1-\varepsilon)}{\frac{w_1}{2}(1-\varepsilon) + \frac{1}{4}(1-\varepsilon)} \right) + p_3 \log\left( \frac{\frac{w_1}{4}(1-\varepsilon) + \frac{\varepsilon}{2}}{\frac{w_1\varepsilon}{2} + \frac{1}{4}(1-\varepsilon)} \right) + p_4 \log \left(\frac{w_1\varepsilon + \frac{\varepsilon}{2}}{\frac{w_1\varepsilon}{2} + \varepsilon} \right) \\
    &= p_2 \log\left(\frac{w_1 + 2}{2w_1 + 1} \right) + p_3 \left( \frac{w_1(1-\varepsilon) + 2\varepsilon}{2w_1\varepsilon + 1-\varepsilon} \right) + p_4 \log \left( \frac{2w_1 + 1}{w_1 + 2} \right) \\
    &= p_3 \left( \frac{w_1(1-\varepsilon) + 2\varepsilon}{2w_1\varepsilon + 1-\varepsilon} \right) + (p_4-p_2) \log \left( \frac{2w_1 + 1}{w_1 + 2} \right).
\end{align*}
We can verify that for $\varepsilon < \frac{1}{3}$, we have $\frac{w_1(1-\varepsilon) + 2\varepsilon}{2w_1\varepsilon + 1-\varepsilon} < \frac{1-\varepsilon}{2\varepsilon}$, thus 
\begin{align*}
    f(i_1 i_2) - f(i_2 i_1) &< p_3 \left( \frac{1-\varepsilon}{2\varepsilon}\right) + (p_4-p_2) \log \left( \frac{2w_1 + 1}{w_1 + 2} \right) = 0\\
    \implies f(i_1 i_2) &< f(i_2 i_1).
\end{align*}
That is, the optimal recommendation list ranks $i_2$ first, then $i_1$ second.

However, we also have
\begin{align*}
    f(i_1) - f(i_2) &= p_2 \log \left( \frac{\frac{w_1}{4}(1-\varepsilon)}{\frac{w_1}{2}(1-\varepsilon) } \right) + p_3 \log\left( \frac{\frac{w_1}{4}(1-\varepsilon)}{\frac{w_1\varepsilon}{2}} \right) + p_4 \log \left(\frac{w_1\varepsilon}{\frac{w_1\varepsilon}{2}} \right) \\
    &= p_2 \log\left(\frac{1}{2} \right) + p_3 \left( \frac{1-\varepsilon}{2\varepsilon} \right) + p_4 \log \left(2\right) \\
    &= p_3 \left( \frac{1-\varepsilon}{2\varepsilon} \right) + (p_4-p_2) \log \left(2 \right).
\end{align*}
Since $w_1 > 1$, we have $\frac{2w_1 + 1}{w_1+2} < 2$, thus
\begin{align*}
    f(i_1) - f(i_2) &> p_3 \left( \frac{1-\varepsilon}{2\varepsilon}\right) + (p_4-p_2) \log \left( \frac{2w_1 + 1}{w_1 + 2} \right) = 0\\
    \implies f(i_1) &> f(i_2).
\end{align*}
That is, the greedy algorithm will first choose $i_1$ instead of $i_2$, thereby constructing a suboptimal list.

Now, we compute $ALG = f(i_1i_2)$ and $OPT = f(i_2i_1)$ for the following set of parameters: $p_1 = 0.05, p_2 = 0.9, p_3=p_4 = 0.025$, $\varepsilon = 10^{-10}$, varying $w_1 > 1$. 

\begin{table}[h!]
  \caption{$ALG$ versus $OPT$ for varying values of $w_1$}
  \label{netflixtable}
  \centering
  \begin{tabular}{crr}
    \toprule
    $w_1$     & $ALG$     & $OPT$ \\
    \midrule
1.1 & -0.823134 & -0.797737 \\
1.5 & -0.691859 & -0.585156 \\
2 & -0.549794 & -0.371873 \\
3.5 & -0.201250 & 0.114023 \\
5 & 0.0311358 & 0.386387 \\
10 & 0.580034 & 1.01213 \\
100 & 2.73099 & 3.20940 \\
    \bottomrule
  \end{tabular}
\end{table}

We now observe that the function does not have consistent sign; $ALG$ and $OPT$ are negative for lower values of $w_1$ and positive for higher values of $w_1$. This is because the $\tilde{q}(g|i)$'s represent a probability distribution and are thus less than $1$, so when the weights are small we take the logarithm of a number less than $1$, so the function is negative; when the weights are sufficiently large, then the inner summand exceeds $1$ and the function becomes positive.

It is unclear how we should think about approximation when the value of a function is not always positive or negative --- for instance, the approximation ratio $ALG/OPT$ is meaningless, especially considering that $ALG$ and $OPT$ may have opposite signs (such as when $w_1 = 3.5$). So if the simple greedy algorithm is not always optimal, but we have no consistent way of comparing its performance with the optimal solution, then it becomes very difficult to understand the maximization (or approximate maximization) of this specific form of the calibration heuristic.

\section{Varying sequential dependencies in calibration}\label{sec:calibration_seqdep}

In Section~\ref{sec:relatedwork}, we described earlier formalisms of sequential submodularity that rely on postfix monotonicity and argued that many natural ordering problems, including the calibration objective function, are not postfix monotone. A different line of papers discussed encodes sequences using DAGs and hypergraphs. Now, we show that this formalism also does not capture the rank-based sequential dependencies that we desire.

We present a simple instance of the calibration problem which hints at the potential intricacies of sequential dependencies. Suppose there are just $2$ genres ($g_1$ and $g_2$), $4$ movies ($i_1$, $i_2$, $i_3$, $i_4$), and $1$ user ($u$). Say that the target distribution is $p(g_1 | u) = p(g_2 | u) = 0.5$, and the weights of the recommended items are $w_1 = 0.5, w_2 = 0.3, w_3 = 0.2$. Suppose further that the movies have genre distributions as follows:

\begin{center}
\begin{tabular}{ p{3cm} p{3cm} }
$p(g_1 | i_1) = 0.4$, & $p(g_2 | i_1) = 0.6$ \\[6pt]
$p(g_1 | i_2) = 0.8$, & $p(g_2 | i_2) = 0.2$ \\[6pt]
$p(g_1 | i_3) = 1$, & $p(g_2 | i_3) = 0$ \\[6pt]
$p(g_1 | i_4) = 0$, & $p(g_2 | i_4) = 1$.
\end{tabular}
\end{center}

Our heuristic for measuring calibration is the overlap measure $G(p,q) = \sum_g \sqrt{p(g|u) \cdot q(g|u)}$. We now consider a few different recommended lists as input to the overlap measure:
\begin{align*}
    f(i_3 i_1 i_2) = G(p, (0.78,0.22)) \approx 0.956 \\
    f(i_3 i_2 i_1) = G(p, (0.82,0.18)) \approx 0.940 \\
    f(i_4 i_1 i_2) = G(p, (0.28,0.72)) \approx 0.974 \\
    f(i_4 i_2 i_1) = G(p, (0.32,0.68)) \approx 0.983 
\end{align*}

Here, we see that $f(i_3 i_1 i_2) > f(i_3 i_2 i_1)$, but $f(i_4 i_1 i_2) < f(i_4 i_2 i_1)$. So it is not always inherently better to rank $i_1$ before $i_2$ or $i_2$ before $i_1$; the optimal ordering is dependent on the context of the rest of the recommended list. Thus this very natural problem setting cannot be satisfactorily encoded by the DAG or hypergraph models of~\cite{Tschiatschek2017} and~\cite{Mitrovic18}.

\section{Approximate optimization with noisy parameters}\label{sec:noisy_approx}
\def\eps{\varepsilon}
\def\hatp{\hat{p}}
\def\hatw{\hat{w}}
\def\hatf{\hat{f}}

In our optimization problem with discrete genres, a user has
a target probability $p(g)$ for each genre $g$, and a weight
$w_i$ that they place on position $i$ in a sequence of recommendations.
An assignment of genres to slots in the recommendation list of length $k$ is
represented by a sequence $S$ of length $k$, where
$s_i = g$ means that the $i^{\rm th}$ position in the list is assigned 
genre $g$.
We seek to maximize the objective function 
\begin{equation*}
f(S) = \sum_{g} \sqrt{p(g)} \sqrt{\sum_{i\in [k] : s_i = g} w_i}.  
\end{equation*}

Now, suppose we only know the user's genre probabilities and 
decaying attention weights approximately; we have $\hatp(g)$ as an 
approximate value for $p(g)$, and we have $\hatw_i$ as an approximate
value for $w_i$.
Suppose these are approximate in the following sense: for some small
positive constant $\eps > 0$, we have
$$\frac{p(g)}{(1 + \eps)} \leq \hatp(g) \leq (1 + \eps) p(g)$$
for all $g$, 
and similarly
$$\frac{w_i}{(1 + \eps)} \leq \hatw_i \leq (1 + \eps) w_i$$
for all $i$.

Given these approximate parameters, suppose we try to optimize with them;
then we are in fact optimizing the function 
$$\hatf(S) = \sum_{g} \sqrt{\hatp(g)} \sqrt{\sum_{i\in [k] : s_i = g} \hatw_i}.$$

We now show that an approximately optimal solution with respect to $\hatf$
is also approximately optimal (with a slightly worse guarantee) with
respect to $f$.
To see this, first 
observe that for any sequence $S$, we have
\begin{eqnarray*}
\hatf(S) & = & \sum_{g} \sqrt{\hatp(g)} \sqrt{\sum_{i\in [k] : s_i = g} \hatw_i} \\
& \leq & \sum_{g} \sqrt{(1 + \eps) p(g)} \sqrt{\sum_{i\in [k] : s_i = g} (1 + \eps) w_i} \\
& = & \sum_{g} \sqrt{(1 + \eps) p(g)} \sqrt{(1 + \eps) \sum_{i\in [k] : s_i = g} w_i} \\
& = & (1 + \eps) \sum_{g} \sqrt{p(g)} \sqrt{\sum_{i\in [k] : s_i = g} w_i} \\
& = & (1 + \eps) f(S),
\end{eqnarray*}
from which it follows that 
\begin{equation}
f(S) \geq \frac{1}{(1 + \eps)} \hatf(S),
\label{eq:f}
\end{equation}
and similarly 
\begin{eqnarray*}
\hatf(S) & = & \sum_{g} \sqrt{\hatp(g)} \sqrt{\sum_{i\in [k] : s_i = g} \hatw_i} \\
& \geq & \sum_{g} \sqrt{\frac{p(g)}{(1 + \eps)}} \sqrt{\sum_{i\in [k] : s_i = g} \frac{w_i}{(1 + \eps)}} \\
& = & \sum_{g} \sqrt{\frac{p(g)}{(1 + \eps)}} \sqrt{\frac{1}{(1 + \eps)} \sum_{i\in [k] : s_i = g} w_i} \\
& = & \frac{1}{(1 + \eps)} \sum_{g} \sqrt{p(g)} \sqrt{\sum_{i\in [k] : s_i = g} w_i} \\
& = & \frac{1}{(1 + \eps)} f(S),
\end{eqnarray*}
which we summarize as 
\begin{equation}
\hatf(S) \geq \frac{1}{(1 + \eps)} f(S).
\label{eq:hatf}
\end{equation}

Now, let $S^*$ be a sequence that optimizes $f$, and let 
$S^o$ be a sequence that optimizes $\hatf$.
For some $\alpha < 1$, 
suppose we use an $\alpha$-approximation algorithm with respect to the
data we have (which serves to define $\hatf$), 
obtaining a solution $S'$ that satisfies the guarantee
$$\hatf(S') \geq \alpha \hatf(S^o).$$
Using the inequalities derived above, we now have
$$f(S')  \geq  \frac{1}{(1 + \eps)} \cdot \hatf(S') 
 \geq  \frac{\alpha}{(1 + \eps)} \cdot \hatf(S^o) 
 \geq  \frac{\alpha}{(1 + \eps)} \cdot \hatf(S^*) 
 \geq  \frac{\alpha}{(1 + \eps)^2} \cdot f(S^*),
$$
where the first inequality follows from (\ref{eq:f}), the second
inequality follows from the $\alpha$-approximation guarantee,
the third inequality follows from the optimality of $S^o$ for
the function $\hatf$, and the fourth inequality follows from
(\ref{eq:hatf}).

It follows that if we have an $\alpha$-approximation with respect
to a set of parameters that are estimated to within a multiplicative
error of $(1 + \eps)$ in each direction, then the resulting solution
is an $\displaystyle{\frac{\alpha}{(1 + \eps)^2}}$-approximation
with respect to the true optimization function.

\section{Computational experiments for the greedy algorithm}\label{sec:comp_experiments}
We computationally investigated the performance of the standard greedy algorithm, measured by the squared Hellinger overlap, across randomly generated problem instances of both the distributional and discrete models (user preferences, movie genres, and position weights). Across varying numbers of slots, movies, and genres, the greedy algorithm consistently performed very close to optimal; below, we present the average- and worst-case $ALG/OPT$ approximation ratios for each model across $N=10000$ trials of each parameter setting.

\begin{table}[h!]

\caption{$ALG/OPT$ for numerical simulations of the greedy algorithm, $N=10000$ trials}
  \label{exptable}

\centering

\begin{tabular}{|c|c|c|rr|rr|}
\hline
\multirow{2}{*}{\textbf{\# slots}} & \multirow{2}{*}{\textbf{\# movies}} & \multirow{2}{*}{\textbf{\# genres}} & \multicolumn{2}{c|}{\textbf{Distributional $ALG/OPT$}}                        & \multicolumn{2}{c|}{\textbf{Discrete $ALG/OPT$}}                              \\ \cline{4-7} 
                                   &                                     &                                     & \multicolumn{1}{c|}{\textbf{Average}} & \multicolumn{1}{c|}{\textbf{Worst}} & \multicolumn{1}{c|}{\textbf{Average}} & \multicolumn{1}{c|}{\textbf{Worst}} \\ \hline
2                                  & 3                                   & 3                                   & \multicolumn{1}{r|}{0.999597}         & 0.926542                            & \multicolumn{1}{r|}{1.0}              & 1.0                                 \\ \hline
2                                  & 3                                   & 4                                   & \multicolumn{1}{r|}{0.999516}         & 0.953029                            & \multicolumn{1}{r|}{1.0}              & 1.0                                 \\ \hline
2                                  & 3                                   & 5                                   & \multicolumn{1}{r|}{0.999580}         & 0.965402                            & \multicolumn{1}{r|}{1.0}              & 1.0                                 \\ \hline
2                                  & 4                                   & 3                                   & \multicolumn{1}{r|}{0.999358}         & 0.950738                            & \multicolumn{1}{r|}{1.0}              & 1.0                                 \\ \hline
2                                  & 4                                   & 4                                   & \multicolumn{1}{r|}{0.999280}         & 0.931698                            & \multicolumn{1}{r|}{1.0}              & 1.0                                 \\ \hline
2                                  & 4                                   & 5                                   & \multicolumn{1}{r|}{0.999342}         & 0.955211                            & \multicolumn{1}{r|}{1.0}              & 1.0                                 \\ \hline
3                                  & 4                                   & 3                                   & \multicolumn{1}{r|}{0.999418}         & 0.950166                            & \multicolumn{1}{r|}{0.998737}         & 0.951152                            \\ \hline
3                                  & 4                                   & 4                                   & \multicolumn{1}{r|}{0.999403}         & 0.966373                            & \multicolumn{1}{r|}{0.999096}         & 0.948820                            \\ \hline
3                                  & 4                                   & 5                                   & \multicolumn{1}{r|}{0.999423}         & 0.968850                            & \multicolumn{1}{r|}{0.999338}         & 0.949216                            \\ \hline
3                                  & 5                                   & 3                                   & \multicolumn{1}{r|}{0.999353}         & 0.957874                            & \multicolumn{1}{r|}{0.998893}         & 0.948382                            \\ \hline
3                                  & 5                                   & 4                                   & \multicolumn{1}{r|}{0.999275}         & 0.957349                            & \multicolumn{1}{r|}{0.999376}         & 0.944912                            \\ \hline
3                                  & 5                                   & 5                                   & \multicolumn{1}{r|}{0.999219}         & 0.968183                            & \multicolumn{1}{r|}{0.999563}         & 0.951208                            \\ \hline
3                                  & 6                                   & 3                                   & \multicolumn{1}{r|}{0.999340}         & 0.972373                            & \multicolumn{1}{r|}{0.999090}         & 0.948851                            \\ \hline
3                                  & 6                                   & 4                                   & \multicolumn{1}{r|}{0.999155}         & 0.963547                            & \multicolumn{1}{r|}{0.999509}         & 0.946946                            \\ \hline
3                                  & 6                                   & 5                                   & \multicolumn{1}{r|}{0.999104}         & 0.958884                            & \multicolumn{1}{r|}{0.999760}         & 0.956849                            \\ \hline
4                                  & 5                                   & 3                                   & \multicolumn{1}{r|}{0.999395}         & 0.970426                            & \multicolumn{1}{r|}{0.997124}         & 0.932464                            \\ \hline
4                                  & 5                                   & 4                                   & \multicolumn{1}{r|}{0.999357}         & 0.965723                            & \multicolumn{1}{r|}{0.997667}         & 0.941692                            \\ \hline
4                                  & 5                                   & 5                                   & \multicolumn{1}{r|}{0.999387}         & 0.976056                            & \multicolumn{1}{r|}{0.998221}         & 0.942862                            \\ \hline
4                                  & 6                                   & 3                                   & \multicolumn{1}{r|}{0.999393}         & 0.968958                            & \multicolumn{1}{r|}{0.996978}         & 0.941517                            \\ \hline
4                                  & 6                                   & 4                                   & \multicolumn{1}{r|}{0.999310}         & 0.963804                            & \multicolumn{1}{r|}{0.997977}         & 0.941866                            \\ \hline
4                                  & 6                                   & 5                                   & \multicolumn{1}{r|}{0.999290}         & 0.963948                            & \multicolumn{1}{r|}{0.998682}         & 0.939017                            \\ \hline
4                                  & 7                                   & 3                                   & \multicolumn{1}{r|}{0.999395}         & 0.967712                            & \multicolumn{1}{r|}{0.996996}         & 0.943542                            \\ \hline
4                                  & 7                                   & 4                                   & \multicolumn{1}{r|}{0.999281}         & 0.972192                            & \multicolumn{1}{r|}{0.998236}         & 0.949526                            \\ \hline
4                                  & 7                                   & 5                                   & \multicolumn{1}{r|}{0.999212}         & 0.978790                            & \multicolumn{1}{r|}{0.998887}         & 0.950546                            \\ \hline
4                                  & 8                                   & 3                                   & \multicolumn{1}{r|}{0.999405}         & 0.967539                            & \multicolumn{1}{r|}{0.996947}         & 0.935040                            \\ \hline
4                                  & 8                                   & 4                                   & \multicolumn{1}{r|}{0.999214}         & 0.973646                            & \multicolumn{1}{r|}{0.998373}         & 0.947801                            \\ \hline
4                                  & 8                                   & 5                                   & \multicolumn{1}{r|}{0.999172}         & 0.979681                            & \multicolumn{1}{r|}{0.999183}         & 0.937648                            \\ \hline
5                                  & 6                                   & 3                                   & \multicolumn{1}{r|}{0.999419}         & 0.975819                            & \multicolumn{1}{r|}{0.995778}         & 0.937196                            \\ \hline
5                                  & 6                                   & 4                                   & \multicolumn{1}{r|}{0.999370}         & 0.975623                            & \multicolumn{1}{r|}{0.996337}         & 0.948621                            \\ \hline
5                                  & 6                                   & 5                                   & \multicolumn{1}{r|}{0.999400}         & 0.979225                            & \multicolumn{1}{r|}{0.996893}         & 0.938278                            \\ \hline
5                                  & 7                                   & 3                                   & \multicolumn{1}{r|}{0.999435}         & 0.981923                            & \multicolumn{1}{r|}{0.995790}         & 0.939833                            \\ \hline
5                                  & 7                                   & 4                                   & \multicolumn{1}{r|}{0.999348}         & 0.979790                            & \multicolumn{1}{r|}{0.996356}         & 0.940764                            \\ \hline
5                                  & 7                                   & 5                                   & \multicolumn{1}{r|}{0.999344}         & 0.982033                            & \multicolumn{1}{r|}{0.997167}         & 0.948319                            \\ \hline
5                                  & 8                                   & 3                                   & \multicolumn{1}{r|}{0.999456}         & 0.979707                            & \multicolumn{1}{r|}{0.995950}         & 0.941829                            \\ \hline
5                                  & 8                                   & 4                                   & \multicolumn{1}{r|}{0.999309}         & 0.979906                            & \multicolumn{1}{r|}{0.996429}         & 0.939618                            \\ \hline
5                                  & 8                                   & 5                                   & \multicolumn{1}{r|}{0.999301}         & 0.978087                            & \multicolumn{1}{r|}{0.997418}         & 0.945195                            \\ \hline
5                                  & 9                                   & 3                                   & \multicolumn{1}{r|}{0.999473}         & 0.970947                            & \multicolumn{1}{r|}{0.996017}         & 0.942974                            \\ \hline
5                                  & 9                                   & 4                                   & \multicolumn{1}{r|}{0.999334}         & 0.981755                            & \multicolumn{1}{r|}{0.996473}         & 0.953384                            \\ \hline
5                                  & 9                                   & 5                                   & \multicolumn{1}{r|}{0.999258}         & 0.975289                            & \multicolumn{1}{r|}{0.997628}         & 0.948620                            \\ \hline
5                                  & 10                                  & 3                                   & \multicolumn{1}{r|}{0.999504}         & 0.983135                            & \multicolumn{1}{r|}{0.996023}         & 0.946908                            \\ \hline
5                                  & 10                                  & 4                                   & \multicolumn{1}{r|}{0.999294}         & 0.970452                            & \multicolumn{1}{r|}{0.996484}         & 0.948485                            \\ \hline
5                                  & 10                                  & 5                                   & \multicolumn{1}{r|}{0.999249}         & 0.977041                            & \multicolumn{1}{r|}{0.997904}         & 0.952026                            \\ \hline
\end{tabular}
\end{table}


\end{document}